\documentclass[a4paper]{article}
\usepackage{amsmath,amsthm,amssymb}
\usepackage{doi}
\usepackage{float}
\usepackage{nicefrac}
\usepackage{placeins}
\usepackage{pgfplots}
\pgfplotsset{compat=1.15}
\usepackage{mathrsfs}
\usetikzlibrary{arrows}
\usepackage{todonotes}
\usepackage[dvipsnames]{xcolor}
\usepackage{complexity}
\usepackage{xspace}
\usepackage[capitalise]{cleveref}
\usepackage{subcaption}
\usepackage{authblk}
\usepackage{thm-restate}
\usepackage{fullpage}

\newtheorem{theorem}{Theorem}
\newtheorem{corollary}[theorem]{Corollary}
\newtheorem{definition}[theorem]{Definition}
\newtheorem{lemma}[theorem]{Lemma}
\newtheorem{claim}[theorem]{Claim}

\newcommand{\felix}[1]{\todo[linecolor=cyan,backgroundcolor=cyan!25,bordercolor=cyan]{F:~#1}}

\newcommand{\RM}{\textsc{Rainbow Matching}\xspace}
\newcommand{\DCS}{\textsc{Maximum $(l, u)$-Matching}\xspace}
\newcommand{\DCSshort}{\textsc{Maximum $(l, u)$-Matching}\xspace}

\tikzset{               
        invisible/.style={opacity=0},
        visible on/.style={alt={#1{}{invisible}}},
        alt/.code args={<#1>#2#3}{%
          \alt<#1>{\pgfkeysalso{#2}}{\pgfkeysalso{#3}} 
        },
        cross/.style={cross out, thick,draw=black, minimum size=2*(#1-\pgflinewidth), inner sep=0pt, outer sep=0pt},
        cross/.default={0.25em},
        edge/.style={ultra thick,black},
        arc/.style={very thick,black,->},
        medge/.style={decorate,very thick,decoration={snake}},
        aedge/.style={very thick,dashed,black},
        dedge/.style={thick,->},
        availedge/.style={thick,blue},
        vertex/.style={inner sep=0.25em,shape=circle,thick,draw,node distance=4em},
        smalledge/.style={thick,DodgerBlue},
        smallvertex/.style={inner sep=0.1em,shape=circle,thick,draw=Black,node distance=4em}
}  

\title{A Complexity Dichotomy for Generalized Rainbow Matchings Based on Color Classes} 
\date{}

\author[1]{Felix Hommelsheim\footnote{Corresponding author}}
\author[2]{Pia Jehmlich}
\author[3]{Moritz Mühlenthaler}

\affil[1]{University of Cologne, Germany, hommelsheim@cs.uni-koeln.de, ORCID: 0000-0003-4444-9793}

\affil[2]{Freie Universität Berlin, Germany, p.jehmlich@fu-berlin.de}

\affil[3]{G-SCOP, Universit\'{e} Grenoble-Alpes, France, moritz.muhlenthaler@grenoble-inp.fr, ORCID: 0000-0002-2729-127X}

\begin{document}

\pagenumbering{gobble}
\maketitle

\begin{abstract}
Given an edge-colored graph, the Maximum Rainbow Matching problem asks for a maximum-cardinality matching of the  graph that contains at most one edge from each color. We provide the following complexity dichotomy for this problem based on the structure of the color classes: Maximum Rainbow Matching admits a polynomial-time algorithm if almost every color class is a complete multipartite graph and it is \NP-hard otherwise. 

To prove the \NP-hardness-part of the dichotomy, we first show that the problem remains \NP-hard even if every color class is a subgraph on four vertices that is either a matching of size two, a path on four vertices or a paw. We then leverage this result to all color classes that are not complete multipartite graphs. For this purpose, we introduce color-closed graph classes, which seem to be an appropriate notion for obtaining complexity classifications for rainbow problems and may be of independent interest. To prove the positive part of the dichotomy, we show that the problem essentially reduces to computing a maximum $(l, u)$-matching, where we heavily exploit that almost all color classes are complete multipartite graphs. In the case where all color classes are complete multipartite, we provide a polynomial-time algorithm that computes a maximum matching containing at most $m_i$ edges from each color class $i$.
\end{abstract}

\newpage

\pagenumbering{arabic}

\section{Introduction}

Given a combinatorial optimization problem with a set of feasible solutions $\mathcal{S} \subseteq 2^X$, a rainbow version of the problem assigns a color to each element of the ground set $X$ and asks for an optimal solution~$S \in \mathcal{S}$ under the additional restriction that each color appears at most once.
Rainbow versions of classical problems that admit efficient algorithms, such as Maximum Matching or $s$–$t$ Connectivity, are known to be \NP-hard, even if each color appears at most twice in the ground set~\cite{Complexity_results_for_rainbow_matchings, diameter2}. A notable exception is the rainbow version of the Spanning Tree problem, which can be represented as the intersection of a graphic matroid and a partition matroid, and therefore admits a polynomial-time algorithm~\cite{schrijver}. 

Rainbow versions of combinatorial optimization problems naturally model requirements related to fairness and diversity that are important both in theory and in practice. 
For example, in social networks, where vertices represent individuals labeled by demographic groups (colors) and edges represent conflicts, a rainbow independent set is a conflict-free group of representatives with at most one member from each demographic.
Furthermore, in communication networks, channels (edges) may be labeled by frequency bands, and one may want to construct a path that avoids reusing the same frequency band in order to reduce the risk of interference; this corresponds to the rainbow $s$-$t$ Connectivity problem. 
Finally, we may consider planning activities that can be carried out in pairs. The edges of the graph represent participants compatibility and the colors represent the different activities. For each activity $i$, there is a budget: it is available at most $m_i$ times. 
The goal is to maximize the number of activities carried out under budget and compatibility constraints;
this is a maximum rainbow matching if all~$m_i = 1$.

From a theoretical perspective, rainbow problems are interesting because they connect several classical topics in combinatorics and algorithms.
For instance, they generalize transversals of latin squares and in extremal combinatorics, they are closely related to anti-Ramsey theory, which studies how large a structure can be while avoiding rainbow substructures; see, e.g.,~\cite{aharoni2009rainbow,drisko1998transversals, hatami2008lower}.
Rainbow constraints correspond to adding a partition matroid on top of some underlying combinatorial optimization problem, which makes them interesting in the context of computational complexity, approximation algorithms, and parameterized algorithms~\cite{gupta2019parameterized, gupta2020quadratic,Complexity_results_for_rainbow_matchings,matching-complete-graph}.

In this paper, we study rainbow versions of the Maximum Matching problem. For this purpose, we consider edge-colored graphs $(V, E; c)$, where $(V, E)$ is an undirected loopless graph that may contain parallel edges and $c : E \to \{1, 2, \ldots, k\}$ is a mapping that assigns to each edge  $E$ a \emph{color} in $\{1, 2, \ldots, k\}$.
A \emph{matching} $M \subseteq E$ of $(V, E)$ is a set of edges, such that no two edges in $M$ share a vertex.
A matching $M$ is \emph{$c$-rainbow} if all edges in $M$ have distinct colors according to $c$. We will just say that $M$ is \emph{rainbow} if the edge-coloring is clear from the context.
Given an edge-colored graph $G = (V, E; c)$ and a number $p \in \mathbb{N}$, the problem \RM asks whether $G$ admits a rainbow matching of size at least~$p$.
The maximization version of the problem, referred to as \textsc{Maximum Rainbow Matching}, asks for a rainbow matching of maximum cardinality of a given edge-colored graph. 

For an edge-colored graph $G = (V, E; c)$ and a color $i$, the subset of edges of color $i$ is called \emph{color class}~$i$. 
A color class is \emph{trivial} if it consists of at most one edge and \emph{non-trivial} otherwise.
For convenience, we may consider each color class $i$ as a subgraph $G_i$ of $G$, whose vertices are those that are incident to an edge of color $i$ and whose edges are precisely those of color~$i$.
We also consider generalizations of the above two problems, where we are additionally given an integer $m_i \geq 1$ for each color class $i$, and the task is to compute a matching $M$ (either of cardinality at least $p$ or of maximum cardinality), such that each color $i$ appears at most $m_i$ times in $M$. We refer to the respective problems as \textsc{Generalized Rainbow Matching} and \textsc{Maximum Generalized Rainbow Matching}.

The current understanding of the complexity of (\textsc{Maximum}) \RM is incomplete, but there are several results for restricted input graphs and for certain restrictions of the color classes. Le and Pfender showed that \RM remains \NP-complete,
even if the input graphs are restricted to paths~\cite{Complexity_results_for_rainbow_matchings}.  
Hence, \textsc{Rainbow Matching} is \NP-complete on graphs of treewidth 1. 
Furthermore, they gave complexity classifications for certain classes of forests in terms of forbidden paths of a certain length as well as
fixed-parameter algorithms along similar lines, where the parameter is the number of connected components~\cite{Complexity_results_for_rainbow_matchings}.
Note that \textsc{Maximum Rainbow Matching} admits a polynomial $(\nicefrac{2}{3} - \varepsilon)$-approximation algorithm by an approximation-preserving reduction to the \textsc{Maximum Independent Set} problem on $K_{1, 4}$-free graphs~\cite{hurkens1989size,Complexity_results_for_rainbow_matchings}.
Further parameterized complexity results can be found in~\cite{gupta2019parameterized, gupta2020quadratic}.
In~\cite{matching-complete-graph}, Paluch and Wasylkiewicz showed that \textsc{Maximum Rainbow Matching} is solvable in polynomial time, if for each color class $i$, the input graph $G$ is a clique on $V(G_i)$ and additionally the non-trivial color classes are almost vertex-disjoint, that is, two non-trivial color classes share at most one vertex. 

\paragraph{Notation.} A \emph{clique} in a graph $G$ is a complete subgraph of $G$ and an independent set is the complement of a clique. 
For an integer $q \geq 1$, we denote by $P_q$ a simple path on $q$ vertices. We let $\overline{G}$ be the complement of a graph $G$.
A graph $G=(V, E)$ is \emph{complete multipartite} (CM) if it admits a partition $\{ V_1, V_2, ..., V_q \}$ of $V$ for some positive integer $q$, such that for $1 \leq i \leq q$, the subgraph $G[V_i]$ induced by $V_i$ is an independent set and for  $1 \leq i < j \leq q$, the subgraph $G[V_i \cup V_j]$ is complete bipartite.
We call the sets $V_1, ..., V_q$ the \emph{parts} of $G$.
A part is \emph{trivial} if it consists of a single vertex and \emph{non-trivial} otherwise.
Notice that the class of CM graphs includes complete (bipartite) graphs.
Furthermore, CM graphs are exactly the $\overline{P}_3$-free graphs, that is, the graphs that do not contain the complement of $P_3$ as induced subgraph.

\subsection{Our Results}

Our contribution is a complete classification of the
complexity of \RM based on the structure of the color
classes. The algorithmic part of the classification can be seen as an extension
of work of Paluch and Wasylkiewicz in~\cite{matching-complete-graph},
who treat the case that the color classes are almost vertex-disjoint cliques. 
We generalize their result to the case that almost all color classes are
complete multipartite graphs, and do not require the color classes to be almost vertex-disjoint.
Furthermore, we show that otherwise
the problem is \NP-hard.  For this purpose, we introduce \emph{color-closed} graph
classes; see \Cref{fig:color-closed} for an illustration. 
This definition is crucial for obtaining a complexity dichotomy for \RM and can be thought of as a natural analogue of a \emph{constraint language} for the Constraint Satisfaction Problem. The notion may be of independent interest for other rainbow problems.

\begin{figure}[t]
    \centering
    \begin{subfigure}[t]{0.2\textwidth} 
        \centering
        \begin{tikzpicture}
            \node[smallvertex] (v1) {};
            \node[smallvertex, right=3em of v1] (v2) {};
            \node[smallvertex, below of=v1] (v3) {};
            \node[smallvertex, right=3em of v3] (v4) {};

            \draw[edge] (v1) to [bend left=25](v3);
            \draw[edge] (v2) to [bend left=25](v4);
            \draw[edge,RoyalBlue] (v1) to [bend right=25](v3);
            \draw[edge,RoyalBlue] (v2) to [bend right=25](v4);
        \end{tikzpicture}
        \caption{\label{fig:ce:a}}
    \end{subfigure}
    \begin{subfigure}[t]{0.2\textwidth} 
        \centering
        \begin{tikzpicture}
            \node[smallvertex] (v1) {};
            \node[smallvertex, right=3em of v1] (v2) {};
            \node[smallvertex, below of=v1] (v3) {};
            \node[smallvertex, right=3em of v3] (v4) {};

            \draw[edge] (v1) to (v3);
            \draw[edge] (v2) to (v4);
            \draw[edge,RoyalBlue] (v1) to (v2);
            \draw[edge,RoyalBlue] (v3) to (v4);
        \end{tikzpicture}
        \caption{\label{fig:ce:b}}
    \end{subfigure}
    \begin{subfigure}[t]{0.25\textwidth} 
        \centering
        \begin{tikzpicture}
            \node[smallvertex] (v1) {};
            \node[smallvertex, right=3em of v1] (v2) {};
            \node[smallvertex, right=3em of v2] (v3) {};
            \node[smallvertex, below of=v1] (v4) {};
            \node[smallvertex, below of=v2] (v5) {};
            \node[smallvertex, below of=v3] (v6) {};

            \draw[edge] (v1) to [bend left=25] (v4);
            \draw[edge] (v3) to (v6);
            \draw[edge,RoyalBlue] (v1) to [bend right=25] (v4);
            \draw[edge,RoyalBlue] (v2) to (v5);
        \end{tikzpicture}
        \caption{\label{fig:ce:c}}
    \end{subfigure}
    \begin{subfigure}[t]{0.25\textwidth} 
        \centering
        \begin{tikzpicture}
            \node[] (v1) {};
            \node[smallvertex, right=3em of v1] (v2) {};
            \node[smallvertex, right=3em of v2] (v3) {};
            \node[smallvertex, below of=v1] (v4) {};
            \node[smallvertex, below of=v2] (v5) {};
            \node[smallvertex, below of=v3] (v6) {};

            \draw[edge] (v2) to (v5);
            \draw[edge] (v3) to (v6);
            \draw[edge,RoyalBlue] (v4) to (v5);
            \draw[edge,RoyalBlue] (v2) to (v3);
        \end{tikzpicture}
        \caption{\label{fig:ce:d}}
    \end{subfigure}
    \caption{Four color-equivalent graphs (see \cref{def:colorclosed}): each graph is composed of two $2K_2$ color classes.} 
 \label{fig:color-closed} 
\end{figure}
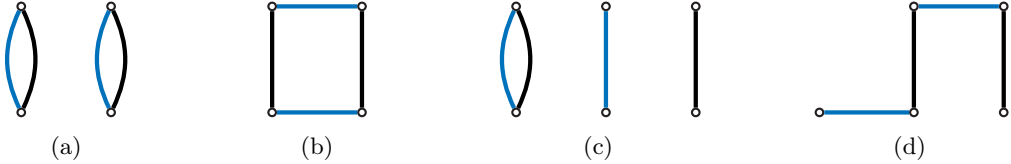

\begin{definition}[color-closed graph class]
    \label{def:colorclosed}
    For an edge-colored graph $G$, let $\eta(G)$ denote the index set of non-trivial color classes.
    Let $G$ and $H$ be edge-colored graphs with color classes $G_1, ..., G_k$ and $H_1, ..., H_{k'}$, respectively.
    Then $G$ and $H$ are \emph{color-equivalent} if there is a bijection $\pi : \eta(G) \to \eta(H)$, such that $G_i$ is isomorphic to $H_{\pi(i)}$ for $i \in \eta(G)$.
    A class $\mathcal{G}$ of edge-colored graphs is \emph{color-closed} if for each $G \in \mathcal{G}$, the class $\mathcal{G}$ contains all edge-colored graphs that are color-equivalent to $G$.
\end{definition}

\begin{definition}[CM-colored graphs] \label{def: F_r freeee}
    Let $\alpha \in \mathbb{N}$ be constant  and let $G = (V, E; c)$ be an edge-colored graph. Then $G$ is \emph{strictly complete-multipartite-colored} (strictly CM-colored) if each color class of $G$ is complete multipartite.
    We say that $G$ is \emph{$\alpha$-CM-colored} if all except at most $\alpha$ color classes are CM graphs. 
    Furthermore, we call a class $\mathcal{G}$ of edge-colored graphs \emph{CM-colored} if there exists some constant $\alpha \in \mathbb{N}$, such that each graph $G \in \mathcal{G}$ is $\alpha$-CM-colored. 
\end{definition}

Our main result is the following complexity dichotomy for \RM.

\begin{restatable}{theorem}{theoremmain}
    \label{thm: dichotomy result rainbow matchings}
    Let $\mathcal{G}$ be a color-closed class of edge-colored graphs. Then \RM restricted to $\mathcal{G}$ admits a polynomial-time algorithm if $\mathcal{G}$ is $\alpha$-CM-colored and it is \NP-complete otherwise.
\end{restatable}

Here is a brief outline of the proof of \Cref{thm: dichotomy result rainbow matchings}. Notice that \RM is in \NP\xspace and let $\mathcal{F} = \{2K_2, P_4, \text{paw}\}$ be the three graphs shown in \Cref{fig: example small graphs}.
These three graphs are the only 4-vertex graphs without isolated vertices that are not CM.
To show the \NP-hardness-part of the dichotomy, we first prove that \RM remains \NP-hard, even if all color classes are pairwise isomorphic and isomorphic to one of the graphs in $\mathcal{F}$.
We then generalize this result to arbitrary color classes that contain graphs in $\mathcal{F}$ as induced subgraphs.
To conclude, we show that every color class (without isolated vertices) that is not a complete multipartite graph contains an induced subgraph isomorphic to one of the graphs in $\mathcal{F}$.

\begin{figure}[t]
    \centering
    \begin{subfigure}[t]{0.30\textwidth} 
        \centering
        \begin{tikzpicture}
            \node[smallvertex] (v1) {};
            \node[smallvertex, right of=v1] (v2) {};
            \node[smallvertex, below of=v1] (v3) {};
            \node[smallvertex, right of=v3] (v4) {};

            \draw[edge] (v1) -- (v2);
            \draw[edge] (v3) -- (v4);
        \end{tikzpicture}
        \caption{$2K_2$}
    \end{subfigure}
    \begin{subfigure}[t]{0.30\textwidth} 
        \centering
        \begin{tikzpicture}
            \node[smallvertex] (v1) {};
            \node[smallvertex, right of=v1] (v2) {};
            \node[smallvertex, below of=v1] (v3) {};
            \node[smallvertex, right of=v3] (v4) {};
            \draw[edge] (v1) -- (v3) -- (v4) -- (v2);
        \end{tikzpicture}
        \caption{$P_4$}
    \end{subfigure}
    \begin{subfigure}[t]{0.30\textwidth} 
        \centering
        \begin{tikzpicture}
            \node[smallvertex] (v1) {};
            \node[smallvertex, right of=v1] (v2) {};
            \node[smallvertex, below of=v1] (v3) {};
            \node[smallvertex, right of=v3] (v4) {};

            \draw[edge] (v3) -- (v2) -- (v1) -- (v3) -- (v4);
        \end{tikzpicture}
        \caption{Paw}
    \end{subfigure}
  \caption{Three graphs on four vertices.}
  \label{fig: example small graphs}
\end{figure}
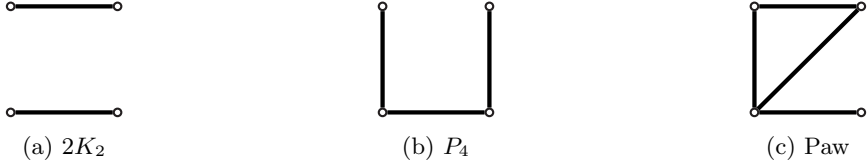

In order to prove the algorithmic part of the dichotomy, we first consider CM-colored graphs. For these graphs, we provide a polynomial-time algorithm for the \textsc{Maximum Generalized Rainbow Matching} problem, where we are allowed to pick at most $m_i$ edges from color class $i$. 
We reduce this problem to the \textsc{Maximum $(l, u)$-Matching} problem, which admits a polynomial-time algorithm~\cite{schrijver} and is defined as follows:
Given a graph $G=(V, E)$ and two functions $l, u: V \rightarrow \mathbb{N}$ such that $l(v) \leq u(v)$ for all $v \in V$, the problem \textsc{Maximum $(l, u)$-Matching}
asks for an edge set $X \subseteq E$ of maximum cardinality, such that for every vertex $v \in V$ we have $l(v) \leq \deg_X(v) \leq u(v)$, where $\deg_X(v)$ denotes the number of edges of $X$ incident to $v$.
The functions $l$ and $u$ are called lower and upper degree bounds, respectively.
We give a polynomial-time algorithm for \textsc{Maximum Generalized Rainbow Matching} on strictly CM-colored graphs, as summarized in \cref{thm:dichotomy-result-rainbow-l-u-matchings}.
\begin{restatable}{theorem}{theoremmaintwo}
    \label{thm:dichotomy-result-rainbow-l-u-matchings}
    \textsc{Generalized Rainbow Matching} on strictly CM-colored graphs admits a polynomial-time algorithm.
\end{restatable}

Combined with a simple enumeration argument for the at most $\alpha$ color classes that are not CM and the hardness-results for \RM, we obtain \Cref{thm: dichotomy result rainbow matchings}.
Note that the analogous statement of \Cref{thm: dichotomy result rainbow matchings} also holds for the maximization version of the problem. 
In the light of \cref{thm:dichotomy-result-rainbow-l-u-matchings}, we believe it would be interesting to explore the setting, where the edges that we may pick from each color class are constrained by a more general matroid than a partition matroid.

\subsection{Related Work}

Adding rainbow constraints to a combinatorial optimization problem imposes the constraints of a partition matroid, where each color class is a part from which we may select at most one item. Hence, optimizing over a matroid under rainbow constraints is equivalent to optimizing over the intersection of two matroids, which can be done in polynomial time~\cite{schrijver}.
However, computing a maximum rainbow forest such that each connected component is rainbow, i.e., contains at most one edge from each color class, is \NP-hard~\cite{rainbow_spanning_forest}.
Further results related to rainbow spanning trees can be found in~\cite{broersma1997spanning}.
Deciding whether a graph contains a rainbow path has been studied by Alon et al.~in~\cite{alon1995color}.
The related rainbow $s$-$t$-connectivity problem has been studied in~\cite{diameter2}, presenting complexity and FPT results; see also~\cite{chen2011complexity}.
The rainbow $s$-$t$-cut problem is known to be \NP-hard on general graphs~\cite{bai2018more} and remains \NP-hard even on planar graphs~\cite{grobler2024solution}. 
We also refer the reader to the following survey for more information on rainbow problems~\cite{kano2008monochromatic-survey}. 

The work of Le and Pfender~\cite{Complexity_results_for_rainbow_matchings} provides a dichotomy for rainbow matchings for certain classes of forests in terms of forbidden paths. 
We point out that their dichotomy is based on the structure of the whole graph and not the individual color classes.
In particular, in their reduction the color classes can have an arbitrary structure and hence their work does not have any immediate implications for our work.
Similarly, our work does not imply any results of their work.

Imposing general matroid constraints instead of those of a partition matroid has received a lot of attention for submodular function maximization~\cite{calinescu2007maximizing,filmus2012tight}: given a submodular set function and a matroid, the task is to select a basis of the matroid of maximum value.
Several different problems related to matchings with matroid constraints, including classical matching, popular matching, and stable matching have been considered in~\cite{csaji2024solving,kakimura2014matching, kamiyama2016popular, kamiyama2017popular, kamiyama2020stable}.
Furthermore, other combinatorial optimization problems have been considered under matroid constraints~\cite{banik2024cuts,hassin1992approximation}.

There are also many results that establish the existence of large rainbow matchings under certain conditions.
Drisko~\cite{drisko1998transversals} proved that any $2n-1$ matchings of size $n$ in a bipartite graph contain a rainbow matching of size $n$. In~\cite{c2}, some additional results on general graphs are provided.
These results are related to the conjectures of Ryser, Brualdi~\cite{brualdi1991combinatorial} and Stein~\cite{stein1975transversals} on partial transversals of size $n-1$ on $n \times n $ Latin squares (resp., size $n$ if $n$ is odd).
Further results related to similar conjectures in this area can be found in~\cite{c4, c5}.
Other works consider the question of how many colors guarantee a rainbow matching~\cite{aharoni2009rainbow, aharoni2015generalization,glebov2014many,kotlar2014large, clemens2016improved}.

\section{Overview}
\label{sec:overview}

In this section, we give an outline of the proof of our main result, which we restate for convenience.

\theoremmain*

We first show the hardness part of \cref{thm: dichotomy result rainbow matchings} and then the polynomial-time part.
To obtain the hardness part, we first show that \textsc{Rainbow Matching} is
\NP-hard if the input graphs are restricted to those, where each color class
is a $2K_2$ by a reduction from
\emph{3-occurrence 3SAT} (\cref{thm: rainbow matching NP-complete
outerplanar}). We then adapt the reduction to show the \NP-hardness of \textsc{Rainbow Matching} for the
cases 1) that each color class is a $P_4$ (\cref{thm: rainbow matching path length three}), and 2) that each color class is a
paw (\cref{thm: rainbow matching triangle plus edge}). 

\begin{restatable}{lemma}{rmtwoedges}
	\label{thm: rainbow matching NP-complete outerplanar}
	\RM is \NP-hard on edge-colored graphs, where each color class is a $2K_2$. 
\end{restatable}
To prove~\Cref{thm: rainbow matching NP-complete outerplanar}, we present a polynomial-time gadget-based reduction from 3-occurrence 3SAT to \RM. 
From an instance of 3-occurrence 3SAT, we construct a gadget for each variable and a gadget for each clause to obtain an edge-colored graph that is an equivalent instance of \RM. The edge-coloring is chosen in such a way that in a maximum rainbow matching, each clause gadget contributes one edge corresponding to a variable that makes a clause true.
Furthermore, the consistency between variables and literals is ensured by the choice of the color classes.
The gadgets are outerplanar and the ``connections'' between the gadgets are made only via the color classes, which implies that the hardness result holds for \RM on  outerplanar graphs, where each color class is a $2K_2$ (\cref{cor:outerplanar}).
By a small modification of the reduction we obtain the next two lemmas, that show \NP-hardness of the problem, whenever all color classes are $P_4$s or paws. However, the resulting graphs are not outerplanar. 

\begin{restatable}{lemma}{rmpaththree}
	\label{thm: rainbow matching path length three}
	\RM is \NP-hard on edge-colored graphs, where each color class is a $P_4$.
\end{restatable}

\begin{restatable}{lemma}{rmtriangleandedge}
	\label{thm: rainbow matching triangle plus edge}
	\RM is \NP-hard on edge-colored graphs, where each color class is a paw.
\end{restatable}

Using \cref{thm: rainbow matching NP-complete outerplanar,thm: rainbow matching path length three,thm: rainbow matching triangle plus edge}, we then prove the hardness part of \cref{thm: dichotomy result rainbow matchings}, which is summarized in the following lemma.

\begin{restatable}{lemma}{rmhardness}
 \label{lem:hardness-part-for-dichotomy}
    Let $\mathcal{G}$ be a color-closed class of graphs, such that for any $\ell > 0$ there is some $G \in \mathcal{G}$, such that $G$ has at least $\ell$ color classes that contain one of $\{2K_2, P_4, \text{paw}\}$ as induced subgraph. Then \RM on $\mathcal{G}$ is \NP-hard.
\end{restatable}

We prove \Cref{lem:hardness-part-for-dichotomy} by a reduction from \RM restricted to graphs in which every color class is either a $2K_2$, $P_4$ or a paw, which is \NP-hard by  \cref{thm: rainbow matching NP-complete outerplanar,thm: rainbow matching path length three,thm: rainbow matching triangle plus edge}.
By noticing that CM graphs are characterized by the absence of $\overline{P}_3$ as induced subgraph, it is easy to prove that any graph without isolated vertices that is not CM contains one of the graphs $\mathcal{F} = \{2K_2, P_4, \text{paw}\}$ as induced subgraph.

\begin{restatable}{lemma}{rmcolorclasses} \label{thm: rainbow matching color classes}
    Let $G$ be a graph without isolated vertices that is not complete multipartite. Then $G$ contains an induced subgraph isomorphic to one of $\{2K_2, P_4, \text{paw}\}$.
\end{restatable}
\begin{proof}
    Recall that complete multipartite graphs are exactly the graphs without $\overline{P}_3$ as induced subgraph. Therefore, any graph without isolated vertices on at most three vertices is complete multipartite. We may therefore assume that $G$ has at least four vertices and contains three vertices $u, v, w$ that induce~$\overline{P}_3$. Assume without loss of generality, that $uv$ is an edge of $G$ and neither $u$ nor $v$ is a neighbor of $w$. Since $G$ contains no isolated vertices, the vertex $w$ has a neighbor $x$ different from $u$ and $v$. The vertex $x$ may be adjacent to none, one, or both of $u$ and $v$. Then the subgraph of $G$ induced by $\{u, v, w, x\}$ is isomorphic to $2K_2$ in the first case, to $P_4$ in the second case and to a paw in the third case.
\end{proof}

As a consequence, graphs that are not CM contain as an induced subgraph one of the color class structures that make \RM\xspace \NP-hard (\cref{thm: rainbow matching NP-complete outerplanar,thm: rainbow matching path length three,thm: rainbow matching triangle plus edge}).
We use this result to differentiate between the polynomial and the \NP-hard cases.
It remains to show that \RM admits a polynomial-time algorithm on $\alpha$-CM-colored graphs for any constant integer~$\alpha \geq 0$. Intuitively, in this graph class, the number of occurrences of color classes containing an induced graph from $\mathcal{F}$ is bounded.

\begin{restatable}{lemma}{rmfrfree}
\label{thm: dichotomy result rainbow matchings poly cases}
    Let $\alpha \in \mathbb{N}$ be a constant. Then \textsc{Maximum Rainbow Matching} on $\alpha$-CM-colored graphs admits a polynomial-time algorithm.
\end{restatable}

We prove this lemma in two steps. 
We consider the problem \textsc{Maximum Generalized Rainbow Matching} and show the following.

\theoremmaintwo*

The proof of \cref{thm:dichotomy-result-rainbow-l-u-matchings} is essentially a reduction to \textsc{Maximum $(l, u)$-Matching} that is inspired by~\cite{matching-complete-graph}.
We first subdivide every edge twice using two subdivision vertices.
Then, for each color class $G_i$ and each non-trivial part of $V(G_i)$, we introduce a new \emph{local} vertex. 
We connect every subdivision vertex that is adjacent to one vertex of the non-trivial part to this local  vertex.
We introduce one additional \emph{universal} vertex for each color class and connect all subdivision vertices to it.
For each newly introduced vertex $v$ we define a carefully chosen degree interval $[l(v), u(v)]$ and solve the corresponding instance of \textsc{Maximum $(l, u)$-Matching} in polynomial time~\cite{schrijver}.
If the resulting $(l, u)$-matching only uses \emph{compatible} pairs of subdivision edges, then we can easily obtain from the $(l, u)$-matching a generalized rainbow matching for the original instance, for which we can also prove optimality. 
However, if the solution contains incompatible pairs of subdivision edges, we first need an additional modification step, which heavily exploits the assumption that each color class is complete multipartite. 
Finally, we use this result for the case in which $\alpha$ color classes are not complete multipartite and all $m_i$ are equal to $1$, i.e., we can select at most one edge from each color class.
Since $\alpha$ is a constant, we can ``guess'' by enumeration a rainbow matching on these color classes that can be extended to a maximum rainbow matching of the input graph.

Using the lemmas above we may prove \cref{thm: dichotomy result rainbow matchings}.

\begin{proof}[Proof of \cref{thm: dichotomy result rainbow matchings}.]
    First, assume that $\mathcal{G}$ is a color-closed $\alpha$-CM-colored class of graphs. Then \textsc{(Maximum) Rainbow Matching} on $\mathcal{G}$ admits a polynomial-time algorithm by \Cref{thm: dichotomy result rainbow matchings poly cases}.
    Assume now that $\mathcal{G}$ is color-closed but not $\alpha$-CM-colored. 
    Then for each $\ell > 0$, there is some $G \in \mathcal{G}$, such that at least $\ell$ color classes are not CM. By \cref{thm: rainbow matching color classes}, every color class of an edge-colored graph that is not CM contains one of the graphs in $\mathcal{F}$ as induced subgraph. 
    We may therefore invoke \cref{lem:hardness-part-for-dichotomy} and conclude that \RM is \NP-hard on $\mathcal{G}$.
\end{proof}

\section{Dichotomy Theorem for Rainbow Matchings}
\label{sec:dichotomy}

In this section, we prove the lemmas used in the proof of  \cref{thm: dichotomy result rainbow matchings}. We first consider the hardness part in \cref{,sec:NP-hard,sec:hardness} and then the polynomial part in \cref{sec:polynomial:mm}.

\subsection{\NP-hard cases}
\label{sec:NP-hard}

In the following three subsections, we show that for each $C \in \mathcal{F} = \{2K_2, P_4, \text{paw}\}$, the problem \RM is \NP-hard on edge-colored graphs, where each color class is isomorphic to $C$. 
The three proofs are loosely based on the \NP-completeness result for rainbow connectivity in outerplanar graphs from \cite{diameter2}.

\subsubsection{Each color class is a matching of size two}
\label{sec:NP:atmosttwo}

We prove the following lemma by a polynomial reduction from the problem 3-occurrence 3SAT, which is known to be \NP-complete~\cite{diameter2} and defined as follows. 

\begin{tabular}{lp{0.75\linewidth}}
    \multicolumn{2}{l}{\textbf{3-occurrence 3SAT}} \\
     \textbf{Input:} & Boolean SAT formula $\phi$ in conjunctive normal form, such that each variable occurs at most three times in $\phi$ and each clause has exactly three literals. \\
     \textbf{Question:} & Is $\phi$ satisfiable?
\end{tabular}

\rmtwoedges*

Let $\phi$ be an instance of 3-occurrence 3SAT with $n$ variables $x_1, \ldots, x_n$ and $m$ clauses $C_1,\ldots, C_m$.
We construct from $\phi$ a graph $G_\phi$ and an edge-coloring $f_\phi$ of~$G_\phi$, as a corresponding instance of \RM. For each $1 \leq j \leq m$ (resp., $1 \leq i \leq n$) we construct a clause gadget $G_j$ (resp., variable gadget $X_i$). For the edge-coloring $f_\phi$, each color class will be a subgraph of a $2K_2$.

Let $1 \leq i \leq n$. The variable gadget $X_i$ is a cycle on the vertices $a_i, u_i, v_i, b_i, \bar{v}_i, \bar{u}_i$, 
see~\cref{fig: var gadget}. 
Let $1 \le j \le m$. The clause gadget $G_j$ is a cycle 
 on six vertices, namely on the vertices $r_{j,1}, r_{j,2}, r_{j,3}, r'_{j,3}, r'_{j,2}, r'_{j,1}$. 
We then add the edge $\{r_{j,1},r'_{j,3}\}$ to $G_j$.
The three edges $\{r_{j,1},r'_{j,1}\}$, $\{r_{j,1},r'_{j,3}\}$ and $\{r_{j,3},r'_{j,3}\}$ correspond to the three literals in the clause $C_j$, see~\cref{fig:clause gadget}. The graph $G_\phi$ is the disjoint union of the gadgets $X_1,X_2,\ldots, X_n$ and $G_1,G_2,\ldots, G_m$.
    The gadgets are not connected to each other.

\begin{figure}[t]
    \centering
    \begin{subfigure}[t]{0.47\textwidth}
        \centering
        \begin{tikzpicture}[small vertex/.append style={node distance=5em}]
            \node[smallvertex,label=above:$u_i$] (v1) {};
            \node[smallvertex,label=above:$v_i$,right of=v1] (v2) {};
            \node[smallvertex,label=right:$b_i$,below right of=v2] (v3) {};
            
            \node[smallvertex,label=below:$\bar v_i$,below left of=v3] (v4) {};
            
            \node[smallvertex,label=below:$\bar u_i$,left of=v4] (v5) {};
            \node[smallvertex,label=left:$a_i$,above left of=v5] (v6) {};

            \draw[edge] (v1) -- node[pos=0.5,label=above:$\bar c_{i, 2}$,yshift=-0.35em] {} (v2);
            \draw[edge] (v2) -- node[pos=0.75,label=above:$c_{i, 3}$,yshift=-0em] {} (v3);
            \draw[edge] (v3) -- node[pos=0.25,label=below:$\bar c_{i, 3}$,yshift=-0em] {} (v4);
            \draw[edge] (v4) -- node[pos=0.5,label=below:$c_{i, 2}$,yshift=0.35em] {} (v5);
            \draw[edge] (v5) -- node[pos=0.75,label=below:$\bar c_{i, 1}$,yshift=0em] {} (v6);
            \draw[edge] (v6) -- node[pos=0.25,label=above:$c_{i, 1}$,yshift=0em] {} (v1);
        \end{tikzpicture}
        \caption{Variable gadget $X_i$ for a variable $x_i$.}
        \label{fig: var gadget}
    \end{subfigure}
    \hfill
    \begin{subfigure}[t]{0.47\textwidth}
        \centering\begin{tikzpicture}[smallvertex/.append style={node distance=6em}]               
        \node[smallvertex,outer sep=0.125em,label=above:$r_{j,1}$] (rj1) {};
        \node[smallvertex,outer sep=0.125em,label=above:$r_{j,2}$,right of=rj1] (rj2) {};
        \node[smallvertex,outer sep=0.125em,label=above:$r_{j,3}$,right of=rj2] (rj3) {};
        \node[smallvertex,outer sep=0.125em,label=below:$r'_{j,1}$,below of=rj1] (sj1) {};
        \node[smallvertex,outer sep=0.125em,label=below:$r'_{j,2}$,right of=sj1] (sj2) {};
        \node[smallvertex,outer sep=0.125em,label=below:$r'_{j,3}$,right of=sj2] (sj3) {};

        \draw[edge,line width=1.1mm,line cap=round] (rj1) -- (rj2) -- (rj3);
        \draw[edge,line width=1.1mm,line cap=round] (sj1) -- (sj2) -- (sj3);
        \draw[edge] (rj1) to node[pos=0.5,label=left:$c_{2,1}$] {} (sj1);
        \draw[edge] (rj3) to node[pos=0.5,label=right:$c_{5,2}$] {} (sj3);
        \draw[edge] (rj1) to node[pos=0.66,label=above:$\bar c_{3,2}$,yshift=-0.25em] {} (sj3);
        
        \end{tikzpicture}
        \caption{For the clause $C_j=(\overline{x}_2 \lor x_3 \lor \overline{x}_5)$ it shows a clause gadget~$C_j$, where $\overline{x}_2$ is the first literal and both $x_3$ and $\overline{x}_5$ are the second literals.}
        \label{fig:clause gadget}
    \end{subfigure}
    \caption{Variable and clause gadgets of the graph $G_\phi$ constructed in~\cref{thm: rainbow matching NP-complete outerplanar}.}
    \label{fig: variable and clause gadget }
\end{figure}
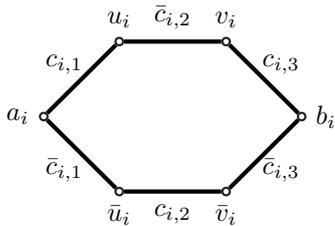
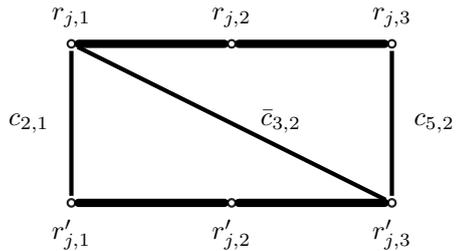

    We now define the coloring $f_\phi$. For $1 \leq i \leq n$, we color the six edges of the variable gadget $X_i$ with new and distinct colors $c_{i,1},\overline{c}_{i,2}, c_{i,3},\overline{c}_{i,3},c_{i,2}$ and $\overline{c}_{i,1}$ consecutively.
    Next, we color the edges of the clause gadgets $C_j$, $1 \leq j \leq m$, by considering the at most three occurrences of each variable $x_i$, $1 \leq i \leq n$, in a clause: 
    Assume that $C_j = l_{j,1} \lor l_{j,2} \lor l_{j,3}$ is the $d$-th clause of $\phi$ in which a literal of $x_i$ appears ($1 \leq d \leq 3$) and that $x_i$ appears in the $k$-th literal of $C_j$ ($k \in \{1, 2, 3\}$), that is, the variable of $l_{j,k}$ is $x_i$.

    If $k \in \{1, 3\}$, we let
    \begin{align} \label{eq: f_phi for k=1,3}
        f_\phi(r_{j,k}, r'_{j,k}) =
        \begin{cases} 
            \overline{c}_{i,d} & \text{if $l_{j,k}$ is positive, and} \\
            c_{i,d} & \text{if $l_{j,k}$ is negative.}
        \end{cases}
    \end{align}
    and if $k = 2$, we let
    \begin{align} \label{eq: f_phi for k=2}
        f_\phi(r_{j,k-1}, r'_{j,k+1}) =
        \begin{cases}
            \overline{c}_{i,d} & \text{if $l_{j,k}$ is positive, and} \\
            c_{i,d} & \text{if $l_{j,k}$ is negative,}
        \end{cases}
    \end{align}

    Let $U = \{ \{r_{j,1},r_{j,2}\},\{r_{j,2},r_{j,3}\} \mid 1 \leq j \leq m \} \cup \{\{r'_{j,1},r'_{j,2}\},\{r'_{j,2},r'_{j,3}\} \mid 1 \leq j \leq m \} $ be the set of remaining edges of $G_\phi$, which we have not yet colored, see the thick edges in \cref{fig:clause gadget}. 
    To each edge in $U$ we assign a unique new color.
    Finally, in order to ensure that each color class induces a $2K_2$, we add new vertices $S = \{s_1, s_2, s_3, s_4\}$ as well as two copies of $\{\{s_1, s_2\}, \{s_3, s_4\}\}$, each one with a unique new color. Then, for each color class in $G_\phi$ that consists of a single edge, we add one more edge $\{s_1, s_3\}$ of the same color.
    This completes the construction of the edge-coloring $f_\phi$ of $G_\phi$.

    Notice that $G_\phi, f_\phi$ can be constructed from $\phi$ in polynomial time and it has the following properties.
    \begin{enumerate}
        \item $G_\phi$ admits a perfect (not necessarily rainbow) matching of size $3n + 3m + 2$, since each variable and each clause gadget admits a perfect matching of size 3 and $G_\phi[S]$ admits a perfect matching of size 2. \label{itm:pm}
        \item Each color class of $G_\phi, f_\phi$ is a matching of size exactly two. Every color class either has an edge in $G_\phi[S]$ or it has one edge in a variable gadget and the other one in a clause gadget. Since the gadgets are vertex-disjoint, two edges having the same color are independent. \label{itm:colorclasses}
        \item Any maximum (rainbow) matching of a clause gadget $G_j, f_\phi$ includes exactly one of the edges $\{r_{j,1},r'_{j,1}\}$, $\{r_{j,1},r'_{j,3}\}$ or $\{r_{j,3},r'_{j,3}\}$. 
        The only way to include two of them is to include $\{r_{j,1},r'_{j,1}\}$ and $\{r_{j,3},r'_{j,3}\}$.
        But then the vertices $r_{j,2}$ and $r'_{j,2}$ are not matched, so the matching is not maximum.
        Recall that these three edges correspond to the three literals of the clause $C_j$. \label{itm:literals}
    \end{enumerate}

    We now define matchings that correspond to the values \textsc{True} and \textsc{False} of the variables~$x_i$. 
    For each variable gadget $X_i$, $1 \leq i \leq n$, we call the edges colored in $c_{i,1},c_{i,2},c_{i,3}$ (resp., $\overline{c}_{i,1},\overline{c}_{i,2},\overline{c}_{i,3}$) of $X_i$ the \emph{$i$-th positive matching} (resp., the \emph{$i$-th negative matching}), see~\cref{fig:pos matching,fig:neg matching}.
    The $i$-th positive matching corresponds to $x_i = \textsc{True}$, and the $i$-th negative matching corresponds to $x_i = \textsc{False}$. 
\cref{thm: rainbow matching NP-complete outerplanar} is a direct consequence of the following lemma and \cref{itm:colorclasses}.

\begin{figure}[t]
        \centering
        \begin{subfigure}[b]{0.45\textwidth}
            \centering
        \begin{tikzpicture}[small vertex/.append style={node distance=5em}]
            \node[smallvertex,label=above:$u_i$] (v1) {};
            \node[smallvertex,label=above:$v_i$,right of=v1] (v2) {};
            \node[smallvertex,label=right:$b_i$,below right of=v2] (v3) {};
            
            \node[smallvertex,label=below:$\bar v_i$,below left of=v3] (v4) {};
            
            \node[smallvertex,label=below:$\bar u_i$,left of=v4] (v5) {};
            \node[smallvertex,label=left:$a_i$,above left of=v5] (v6) {};

            \draw[edge] (v1) -- node[pos=0.5,yshift=-0.35em] {} (v2);
            \draw[edge] (v2) -- node[pos=0.75,label=above:$c_{i, 3}$,yshift=-0em] {} (v3);
            \draw[line width=0.75em,draw opacity=0.5,line cap=round,color=RoyalBlue] (v2) -- (v3);
            \draw[line width=0.75em,draw opacity=0.5,line cap=round,color=RoyalBlue] (v4) -- (v5);
            \draw[line width=0.75em,draw opacity=0.5,line cap=round,color=RoyalBlue] (v6) -- (v1);
            \draw[edge] (v3) -- node[pos=0.25,yshift=-0em] {} (v4);
            \draw[edge] (v4) -- node[pos=0.5,label=below:$c_{i, 2}$,yshift=0.35em] {} (v5);
            \draw[edge] (v5) -- node[pos=0.75,yshift=0em] {} (v6);
            \draw[edge] (v6) -- node[pos=0.25,label=above:$c_{i, 1}$,yshift=0em] {} (v1);
        \end{tikzpicture}
            \caption{The $i$-th positive matching}
            \label{fig:pos matching}
        \end{subfigure}
        \hfill
        \begin{subfigure}[b]{0.45\textwidth}
            \centering
        \begin{tikzpicture}[small vertex/.append style={node distance=5em}]
            \node[smallvertex,label=above:$u_i$] (v1) {};
            \node[smallvertex,label=above:$v_i$,right of=v1] (v2) {};
            \node[smallvertex,label=right:$b_i$,below right of=v2] (v3) {};
            
            \node[smallvertex,label=below:$\bar v_i$,below left of=v3] (v4) {};
            
            \node[smallvertex,label=below:$\bar u_i$,left of=v4] (v5) {};
            \node[smallvertex,label=left:$a_i$,above left of=v5] (v6) {};

            \draw[line width=0.75em,draw opacity=0.5,line cap=round,color=RoyalBlue] (v1) -- (v2);
            \draw[line width=0.75em,draw opacity=0.5,line cap=round,color=RoyalBlue] (v3) -- (v4);
            \draw[line width=0.75em,draw opacity=0.5,line cap=round,color=RoyalBlue] (v5) -- (v6);

            \draw[edge] (v1) -- node[pos=0.5,label=above:$\bar c_{i, 2}$,yshift=-0.35em] {} (v2);
            \draw[edge] (v2) -- node[pos=0.75,yshift=-0em] {} (v3);
            \draw[edge] (v3) -- node[pos=0.25,label=below:$\bar c_{i, 3}$,yshift=-0.25em] {} (v4);
            \draw[edge] (v4) -- node[pos=0.5,yshift=0.35em] {} (v5);
            \draw[edge] (v5) -- node[pos=0.75,label=below:$\bar c_{i, 1}$,yshift=-0.25em] {} (v6);
            \draw[edge] (v6) -- node[pos=0.25,yshift=0em] {} (v1);
        \end{tikzpicture}
            \caption{The $i$-th negative matching}
            \label{fig:neg matching}
        \end{subfigure}
        \caption{The $i$-th positive and negative matching for a variable gadget~$X_i$ of the graph $G_\phi, f_\phi$ constructed in~\cref{thm: rainbow matching NP-complete outerplanar}.}
        \label{fig:posivive and negative matching}
\end{figure}
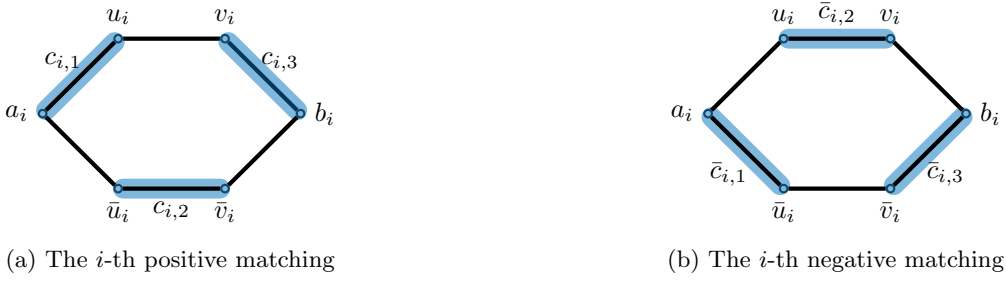

\begin{lemma}
    \label{thm: Hilfslemma zu outerplanar rainbow matchings}
    $G_\phi, f_\phi$ has a rainbow matching of size $3(n+m) + 2$ if and only if $\phi$ is satisfiable.
\end{lemma}
\begin{proof}
        Let $M$ be a perfect matching of $G_\phi$, which has size $3(n+m) + 2$ (see \cref{itm:pm}). For each variable gadget $X_i$, $1 \leq i \leq n$, the matching $M$ restricted to $X_i$ is either the $i$-th positive or the $i$-th negative matching of $X_i$.

        For $1 \leq i \leq n$, we let $x_i = \textsc{True}$ (resp., $x_i = \textsc{False}$) if $M$ restricted to $X_i$ is the $i$-th positive (resp., the $i$-th negative) matching to obtain an assignment $z$.
        It remains to show that if $M$ is rainbow then $z$ satisfies $\phi$. So assume that $M$ is rainbow. Since $M$ is perfect and by~\cref{itm:literals}, we have that $M$ contains exactly one of the edges 
        $\{r_{j,1},r'_{j,1}\}$, $\{r_{j,1},r'_{j,3}\}$ and $\{r_{j,3},r'_{j,3}\}$ of the clause gadget $G_j$ for $1\leq j \leq m$.

        We show that if $M$ contains the edge $e = \{r_{j,1}, r'_{j,1}\}$ then $z$ sets the first literal of the clause $C_j$ to \textsc{True}.
        By~\eqref{eq: f_phi for k=1,3}, the color $f_\phi(e)$ is $\overline{c}_{i,1}$ if the first literal of $C_j$ is positive and $c_{i,1}$ otherwise. Since each color appears at most once in $M$, in the first case, the matching $M$ restricted to the corresponding variable gadget $X_i$ is the $i$-th positive matching and in the latter case it is the $i$-th negative matching.
        The argument for the edges $\{r_{j,3}, r'_{j,3}\}$ and $\{r_{j,1}, r'_{j,3}\}$ is the same (for $\{r_{j,1}, r'_{j,3}\}$, the color is defined in~\eqref{eq: f_phi for k=2}). We conclude that $z$ satisfies~$\phi$.

        To prove the converse direction, let $z$ be a satisfying assignment of $\phi$. We construct from $z$ a rainbow matching $M'$ of $G_\phi, f_\phi$ of size $3(n+m) + 2$. First, we add two independent edges on $G_\phi[S]$ of different colors. Then, for $1 \leq i \leq n$, if $z(x_i) = \textsc{True}$ (resp., $z(x_i) = \textsc{False})$ we add to $M'$ the $i$-th positive (resp., $i$-th negative) matching.
        So far, the matching $M'$ is rainbow and of size $3n+2$.
        
        It remains to show that for each clause gadget $G_j$, $1\leq j\leq m$, we may add to $M'$ exactly one of the edges $\{ \{r_{j,1},r'_{j,1}\},\{r_{j,1},r'_{j,3}\},\{r_{j,3},r'_{j,3}\}\}$ without using a color twice. 
        Let $j \in \{1, 2, \ldots, m\}$. Since $z$ satisfies $\phi$, there is a literal $l$ of $C_j$ that evaluates to \textsc{True} according to~$z$. Let $x_i$ be the variable of $l$ and assume that $C_j$ is the $k$-th clause in which $x$ appears. Furthermore, assume that $l$ is the $y$-th literal of $C_j$. If $l$ is positive, then $f_\phi(r_{j,k}, r'_{j,k}) = \overline{c}_{i,k}$ if $k \in \{1, 3\}$ and $f_\phi(r_{j,1}, r'_{j,3}) = \overline{c}_{i,k}$ if $k=2$, by \eqref{eq: f_phi for k=1,3} and \eqref{eq: f_phi for k=2}. Furthermore, $z(x_i) = \textsc{True}$, so $M'$ restricted to $X_i$ is the $i$-th positive matching. We may therefore add one of the edges $\{ \{r_{j,1},r'_{j,1}\},\{r_{j,1},r'_{j,3}\},\{r_{j,3},r'_{j,3}\}\}$ to $M'$ without using a color twice. We then add two edges of $U$ to $M'$ to obtain a perfect matching on $G_j$. The same argument holds if $l$ is negative and hence, we have that $M'$ restricted to $X_i$ is the $i$-th negative matching. 
        We conclude that, after treating each clause of $\phi$, $M'$ is a perfect rainbow matching of $G_\phi, f_\phi$ of size $3(n+m)+2$.
\end{proof}

Since $G'_\phi$ is an outerplanar graph (i.e., $G'_\phi$ admits a crossing-free drawing in the plane, such that each vertex is on the outer face), we obtain the following corollary.

\begin{corollary}
    \label{cor:outerplanar}
    \RM restricted to outerplanar graphs, where each color class is isomorphic to $2K_2$, is \NP-hard.
\end{corollary}

\subsubsection{Each color class is a path on four vertices}
\label{sec:NP:atmostthree}

We adapt the reduction from \cref{sec:NP:atmosttwo} in order to show the following.

\rmpaththree*

To prove \cref{thm: rainbow matching path length three}, we first construct from a 3-occurrence 3SAT formula $\phi$ the graph $G_\phi, f_\phi$ as described in \cref{sec:NP:atmosttwo}. 
Recall that any color class of $G_\phi, f_\phi$ is a matching of size two. 
For each color class, we add an edge of the same color to form a path on  three edges.
We show that the edges we added do not occur in any maximum (rainbow) matching and may then apply~\cref{thm: Hilfslemma zu outerplanar rainbow matchings}.

We now describe the construction. Let $\phi$ be a 3-occurrence 3SAT formula on $n$ variables $x_1,\ldots, x_n$ and $m$ clauses $C_1, \ldots, C_m$. We construct, in polynomial time, the edge-colored graph $G_\phi, f_\phi$ as described in~\cref{sec:NP:atmosttwo}.
For each clause gadget $G_j$, $1\leq j \leq m$, we connect the edges $\{r_{j,1},r'_{j,1}\}$, $\{r_{j,1},r'_{j,3}\}$ and $\{r_{j,3},r'_{j,3}\}$ to the edges colored with the same color in a variable gadget. 
For each edge in a variable gadget $X_i$, $1\leq i\leq n$, we add edges to $G_\phi, f_\phi$ as follows: 
\begin{itemize}
    \item If the edge $\{a_i, \overline{u}_i\}$ (resp., $\{a_i, u_i\}$) has the same color as either $\{r_{j,1},r'_{j,1}\}$ or $\{r_{j,1},r'_{j,3}\}$, then the edge $\{\overline{u}_i,r_{j,1}\}$ (resp., $\{u_i,r_{j,1}\}$) is added to the graph with the same color.

    \item If the edge $\{a_i, \overline{u}_i\}$ (resp., $\{a_i, u_i\}$) has the same color as $\{r_{j,3},r'_{j,3}\}$, then the edge $\{\overline{u}_i,r_{j,3}\}$ (resp., $\{u_i,r_{j,3}\}$) is added to the graph with the same color.

    \item If the edge $\{\overline{u}_i, \overline{v}_i\}$ (resp., $\{u_i, v_i\}$) has the same color as either $\{r_{j,1},r'_{j,1}\}$ or $\{r_{j,1},r'_{j,3}\}$, then the edge $\{\overline{v}_i,r_{j,1}\}$ (resp., $\{v_i,r_{j,1}\}$) is added to the graph with the same color.

    \item If the edge $\{\overline{u}_i, \overline{v}_i\}$ (resp., $\{u_i, v_i\}$) has the same color as $\{r_{j,3},r'_{j,3}\}$, then the edge $\{\overline{v}_i,r_{j,3}\}$ (resp., $\{v_i,r_{j,3}\}$) is added to the graph with the same color.

    \item If the edge $\{\overline{v}_i, b_i\}$ (resp., $\{v_i, b_i\}$) has the same color as either $\{r_{j,1},r'_{j,1}\}$ or $\{r_{j,1},r'_{j,3}\}$, then the edge $\{b_i,r_{j,1}\}$ (resp., $\{b_i,r_{j,1}\}$) is added to the graph with the same color.

    \item If the edge $\{\overline{v}_i, b_i\}$ (resp., $\{v_i, b_i\}$) has the same color as $\{r_{j,3},r'_{j,3}\}$, then the edge $\{b_i,r_{j,3}\}$ (resp., $\{b_i,r_{j,3}\}$) is added to the graph with the same color.
    \item For each of the two color classes $\{\{s_1, s_2\}, \{s_3, s_4\}\}$, we add the edge $\{s_2, s_3\}$ of the same color.
    \item For each color class of the form $\{u, v\}, \{s_1, s_3\}$, where $u, v \in V(G_\phi)$, we add the edge $\{v, s_1\}$ of the same color.
\end{itemize}

Let $G'_\phi, f'_\phi$ be the resulting graph, see~\Cref{fig:graph gphi path length three} for an illustration, and let $U' = E(G'_\phi) \setminus E(G_\phi)$ be the edges added as described above. 
Notice that each color class of $G'_\phi, f'_\phi$ is a $P_4$. 
Now, \cref{thm: rainbow matching path length three} follows from the next lemma and \cref{thm: Hilfslemma zu outerplanar rainbow matchings}: Since we may delete the edges $U'$ from $G'_\phi$ without altering the optimal rainbow matchings, it suffices to consider $G_\phi, f_\phi$ and we may invoke \cref{thm: Hilfslemma zu outerplanar rainbow matchings}.

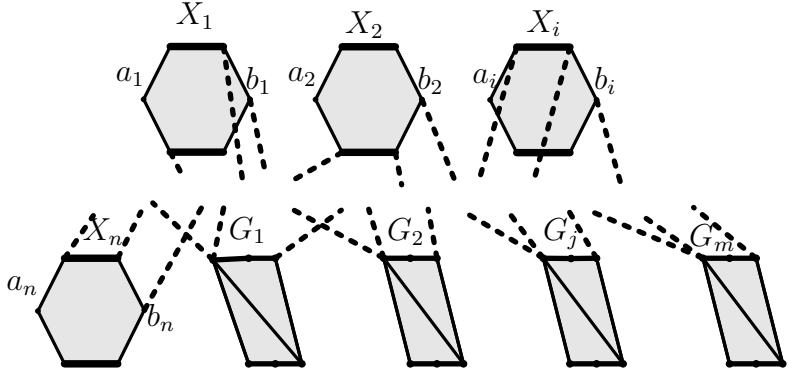
\begin{figure}[t]
    \centering
    \begin{tikzpicture}[line cap=round,line join=round,>=triangle 45,x=1.0cm,y=1.0cm, scale=0.35]
    \tikzstyle{every node}=[font=\large]
    \fill[line width=2.pt,fill=black,fill opacity=0.10000000149011612] (6.973852935290999,12.01973552365099) -- (7.973852935290999,14.01973552365099) -- (9.973852935291,14.01973552365099) -- (10.973852935291,12.01973552365099) -- (9.973852935291,10.01973552365099) -- (7.973852935290999,10.01973552365099) -- cycle;
    \fill[line width=2.8pt,fill=black,fill opacity=0.10000000149011612] (13.436896919503397,12.011591519573805) -- (14.436896919503397,14.011591519573805) -- (16.436896919503393,14.011591519573805) -- (17.436896919503393,12.011591519573805) -- (16.436896919503393,10.011591519573805) -- (14.436896919503397,10.011591519573805) -- cycle;
    \fill[line width=2.pt,fill=black,fill opacity=0.10000000149011612] (20.,12.) -- (21.,14.) -- (23.,14.) -- (24.,12.) -- (23.,10.) -- (21.,10.) -- cycle;
    \fill[line width=2.pt,fill=black,fill opacity=0.10000000149011612] (3.,4.) -- (4.,6.) -- (6.,6.) -- (7.,4.034401754337004) -- (6.,2.) -- (4.,2.) -- cycle;
    \fill[line width=2.pt,fill=black,fill opacity=0.10000000149011612] (9.58527674863891,5.932225525998697) -- (10.92414911864543,6.) -- (11.92414911864543,6.) -- (12.92414911864543,2.) -- (11.92414911864543,2.) -- (10.92414911864543,2.) -- cycle;
    \fill[line width=2.pt,fill=black,fill opacity=0.10000000149011612] (16.,6.) -- (17.,6.) -- (18.,6.) -- (19.,2.) -- (18.,2.) -- (17.,2.) -- cycle;
    \fill[line width=2.pt,fill=black,fill opacity=0.10000000149011612] (22.,6.) -- (23.04211595780728,5.989291962100647) -- (24.,6.) -- (25.,2.) -- (24.,2.) -- (23.,2.) -- cycle;
    \fill[line width=2.pt,fill=black,fill opacity=0.10000000149011612] (28.,6.) -- (29.,6.) -- (30.,6.) -- (31.,2.) -- (30.,2.) -- (29.,2.) -- cycle;
    \draw [line width=1.2pt] (6.973852935290999,12.01973552365099)-- (7.973852935290999,14.01973552365099);
    \draw [line width=2.8pt] (7.973852935290999,14.01973552365099)-- (9.973852935291,14.01973552365099);
    \draw [line width=1.2pt] (9.973852935291,14.01973552365099)-- (10.973852935291,12.01973552365099);
    \draw [line width=1.2pt] (10.973852935291,12.01973552365099)-- (9.973852935291,10.01973552365099);
    \draw [line width=2.8pt] (9.973852935291,10.01973552365099)-- (7.973852935290999,10.01973552365099);
    \draw [line width=1.2pt] (7.973852935290999,10.01973552365099)-- (6.973852935290999,12.01973552365099);
    \draw [line width=1.2pt] (13.436896919503397,12.011591519573805)-- (14.436896919503397,14.011591519573805);
    \draw [line width=2.8pt] (14.436896919503397,14.011591519573805)-- (16.436896919503393,14.011591519573805);
    \draw [line width=1.2pt] (16.436896919503393,14.011591519573805)-- (17.436896919503393,12.011591519573805);
    \draw [line width=1.2pt] (17.436896919503393,12.011591519573805)-- (16.436896919503393,10.011591519573805);
    \draw [line width=2.8pt] (16.436896919503393,10.011591519573805)-- (14.436896919503397,10.011591519573805);
    \draw [line width=1.2pt] (14.436896919503397,10.011591519573805)-- (13.436896919503397,12.011591519573805);
    \draw [line width=1.2pt] (20.,12.)-- (21.,14.);
    \draw [line width=2.8pt] (21.,14.)-- (23.,14.);
    \draw [line width=1.2pt] (23.,14.)-- (24.,12.);
    \draw [line width=1.2pt] (24.,12.)-- (23.,10.);
    \draw [line width=2.8pt] (23.,10.)-- (21.,10.);
    \draw [line width=1.2pt] (21.,10.)-- (20.,12.);
    \draw [line width=1.2pt] (3.,4.)-- (4.,6.);
    \draw [line width=2.8pt] (4.,6.)-- (6.,6.);
    \draw [line width=1.2pt] (6.,6.)-- (7.,4.034401754337004);
    \draw [line width=1.2pt] (7.,4.034401754337004)-- (6.,2.);
    \draw [line width=2.8pt] (6.,2.)-- (4.,2.);
    \draw [line width=1.2pt] (4.,2.)-- (3.,4.);
    \draw [line width=2.pt] (9.58527674863891,5.932225525998697)-- (10.92414911864543,6.);
    \draw [line width=2.pt] (10.92414911864543,6.)-- (11.92414911864543,6.);
    \draw [line width=1.2pt] (11.92414911864543,6.)-- (12.92414911864543,2.);
    \draw [line width=2.pt] (12.92414911864543,2.)-- (11.92414911864543,2.);
    \draw [line width=2.pt] (11.92414911864543,2.)-- (10.92414911864543,2.);
    \draw [line width=1.2pt] (10.92414911864543,2.)-- (9.58527674863891,5.932225525998697);
    \draw [line width=2.pt] (16.,6.)-- (17.,6.);
    \draw [line width=2.pt] (17.,6.)-- (18.,6.);
    \draw [line width=1.2pt] (18.,6.)-- (19.,2.);
    \draw [line width=2.pt] (19.,2.)-- (18.,2.);
    \draw [line width=2.pt] (18.,2.)-- (17.,2.);
    \draw [line width=1.2pt] (17.,2.)-- (16.,6.);
    \draw [line width=2.pt] (22.,6.)-- (23.04211595780728,5.989291962100647);
    \draw [line width=2.pt] (23.04211595780728,5.989291962100647)-- (24.,6.);
    \draw [line width=1.2pt] (24.,6.)-- (25.,2.);
    \draw [line width=2.pt] (25.,2.)-- (24.,2.);
    \draw [line width=2.pt] (24.,2.)-- (23.,2.);
    \draw [line width=1.2pt] (23.,2.)-- (22.,6.);
    \draw [line width=2.pt] (28.,6.)-- (29.,6.);
    \draw [line width=2.pt] (29.,6.)-- (30.,6.);
    \draw [line width=1.2pt] (30.,6.)-- (31.,2.);
    \draw [line width=2.pt] (31.,2.)-- (30.,2.);
    \draw [line width=2.pt] (30.,2.)-- (29.,2.);
    \draw [line width=1.2pt] (29.,2.)-- (28.,6.);
    \draw [line width=1.2pt] (9.58527674863891,5.932225525998697)-- (10.92414911864543,2.);
    \draw [line width=1.2pt] (9.58527674863891,5.932225525998697)-- (12.92414911864543,2.);
    \draw [line width=1.2pt] (11.92414911864543,6.)-- (12.92414911864543,2.);
    \draw [line width=1.2pt] (16.,6.)-- (17.,2.);
    \draw [line width=1.2pt] (16.,6.)-- (19.,2.);
    \draw [line width=1.2pt] (18.,6.)-- (19.,2.);
    \draw [line width=1.2pt] (22.,6.)-- (23.,2.);
    \draw [line width=1.2pt] (22.,6.)-- (25.,2.);
    \draw [line width=1.2pt] (24.,6.)-- (25.,2.);
    \draw [line width=1.2pt] (28.,6.)-- (29.,2.);
    \draw [line width=1.2pt] (28.,6.)-- (31.,2.);
    \draw [line width=1.2pt] (30.,6.)-- (31.,2.);
    \draw [line width=2.pt,loosely dotted] (4.,6.)-- (5.081043043159457,7.6135604414542675);
    \draw [line width=2.pt, loosely dotted] (6.,6.)-- (6.8699342686837195,7.751167458802287);
    \draw [line width=2.pt, loosely dotted] (28.,6.)-- (25.61889038235147,7.682363950128278);
    \draw [line width=2.pt, loosely dotted] (30.,6.)-- (27.648593888234767,7.923176230487312);
    \draw [line width=2.pt, loosely dotted] (7.973852935290999,10.01973552365099)-- (8.452414968185952,9.092835877945479);
    \draw [line width=2.pt, loosely dotted] (10.973852935291,12.01973552365099)-- (11.582974612853413,9.092835877945479);
    \draw [line width=2.pt, loosely dotted] (9.973852935291,14.01973552365099)-- (10.68852900009128,8.989630614934464);
    \draw [line width=2.pt, loosely dotted] (14.436896919503397,10.011591519573805)-- (12.64942899730057,8.989630614934464);
    \draw [line width=2.pt, loosely dotted] (17.436896919503393,12.011591519573805)-- (18.66973600627645,8.95522886059746);
    \draw [line width=2.pt, loosely dotted] (16.436896919503393,10.011591519573805)-- (16.67443425473016,8.783220088912435);
    \draw [line width=2.pt, loosely dotted] (9.58527674863891,5.932225525998697)-- (7.317157075064789,8.095185002172336);
    \draw [line width=2.pt, loosely dotted] (9.58527674863891,5.932225525998697)-- (10.,8.);
    \draw [line width=2.pt, loosely dotted] (11.92414911864543,6.)-- (14.403918468487833,7.819970967476297);
    \draw [line width=2.pt, loosely dotted] (16.,6.)-- (12.305411453930523,7.991979739161322);
    \draw [line width=2.pt, loosely dotted] (16.,6.)-- (15.435971098597983,7.854372721813302);
    \draw [line width=2.pt, loosely dotted] (18.,6.)-- (17.672085130503312,7.854372721813302);
    \draw [line width=2.pt, loosely dotted] (7.,4.034401754337004)-- (9.174851809263062,7.957577984824317);
    \draw [line width=2.pt, loosely dotted] (22.,6.)-- (18.910548286635493,7.819970967476297);
    \draw [line width=2.pt, loosely dotted] (22.,6.)-- (20.665037757822752,7.785569213139293);
    \draw [line width=2.pt, loosely dotted] (24.,6.)-- (22.90115178972808,7.888774476150307);
    \draw [line width=2.pt, loosely dotted] (21.,14.)-- (19.6329851277126,9.127237632282485);
    \draw [line width=2.pt, loosely dotted] (23.,14.)-- (21.628286879258894,8.88642535192345);
    \draw [line width=2.pt, loosely dotted] (24.,12.)-- (24.896453541274376,8.88642535192345);
    \draw [line width=2.pt, loosely dotted] (28.,6.)-- (23.933204419838223,7.716765704465282);
    \begin{scriptsize}
    \draw [fill=black] (6.973852935290999,12.01973552365099) circle (2.5pt);
    \draw[color=black] (6.47431409380816,12.877028855016022) node {$a_1$};
    \draw [fill=black] (7.973852935290999,14.01973552365099) circle (2.5pt);
    \draw [fill=black] (9.973852935291,14.01973552365099) circle (2.5pt);
    \draw [fill=black] (10.973852935291,12.01973552365099) circle (2.5pt);
    \draw[color=black] (11.290559700988867,12.739421837668003) node {$b_1$};
    \draw [fill=black] (9.973852935291,10.01973552365099) circle (2.5pt);
    \draw [fill=black] (7.973852935290999,10.01973552365099) circle (2.5pt);
    \draw[color=black] (8.916838651735517,15.181946395595352) node {$X_1$};
    \draw [fill=black] (13.436896919503397,12.011591519573805) circle (2.5pt);
    \draw[color=black] (12.941843909165108,12.877028855016022) node {$a_2$};
    \draw [fill=black] (14.436896919503397,14.011591519573805) circle (2.5pt);
    \draw [fill=black] (16.436896919503393,14.011591519573805) circle (2.5pt);
    \draw [fill=black] (17.436896919503393,12.011591519573805) circle (2.5pt);
    \draw[color=black] (17.758089516345816,12.739421837668003) node {$b_2$};
    \draw [fill=black] (16.436896919503393,10.011591519573805) circle (2.5pt);
    \draw [fill=black] (14.436896919503397,10.011591519573805) circle (2.5pt);
    \draw[color=black] (15.281163204081452,14.837928852225302) node {$X_2$};
    \draw [fill=black] (20.,12.) circle (2.5pt);
    \draw[color=black] (19.822194776566114,12.808225346342011) node {$a_i$};
    \draw [fill=black] (21.,14.) circle (2.5pt);
    \draw [fill=black] (23.,14.) circle (2.5pt);
    \draw [fill=black] (24.,12.) circle (2.5pt);
    \draw[color=black] (24.328824594713776,12.739421837668003) node {$b_i$};
    \draw [fill=black] (23.,10.) circle (2.5pt);
    \draw [fill=black] (21.,10.) circle (2.5pt);
    \draw[color=black] (22.023907054134437,14.906732360899312) node {$X_i$};
    \draw [fill=black] (4.,6.) circle (2.5pt);
    \draw [fill=black] (6.,6.) circle (2.5pt);
    \draw [fill=black] (7.,4.034401754337004) circle (2.5pt);
    \draw[color=black] (7.634357952233286,3.758214831482858) node {$b_n$};
    \draw [fill=black] (6.,2.) circle (2.5pt);
    \draw [fill=black] (4.,2.) circle (2.5pt);
    \draw [fill=black] (3.,4.) circle (2.5pt);
    \draw[color=black] (2.414907082041566,4.964625357504888) node {$a_n$};
    \draw[color=black] (5.463842165990953,6.99432886338818) node {$X_n$};
    \draw [fill=black] (9.58527674863891,5.932225525998697) circle (3.5pt);
    \draw [fill=black] (10.92414911864543,6.) circle (3.5pt);
    \draw [fill=black] (11.92414911864543,6.) circle (3.5pt);
    \draw [fill=black] (12.92414911864543,2.) circle (3.5pt);
    \draw [fill=black] (11.92414911864543,2.) circle (3.5pt);
    \draw [fill=black] (10.92414911864543,2.) circle (3.5pt);
    \draw[color=black] (10.877738648944806,7.063132372062189) node {$G_1$};
    \draw [fill=black] (16.,6.) circle (3.5pt);
    \draw [fill=black] (17.,6.) circle (3.5pt);
    \draw [fill=black] (18.,6.) circle (3.5pt);
    \draw [fill=black] (19.,2.) circle (3.5pt);
    \draw [fill=black] (18.,2.) circle (3.5pt);
    \draw [fill=black] (17.,2.) circle (3.5pt);
    \draw[color=black] (16.829242149246678,6.99432886338818) node {$G_2$};
    \draw [fill=black] (22.,6.) circle (3.5pt);
    \draw [fill=black] (23.04211595780728,5.989291962100647) circle (3.5pt);
    \draw [fill=black] (24.,6.) circle (3.5pt);
    \draw [fill=black] (25.,2.) circle (3.5pt);
    \draw [fill=black] (24.,2.) circle (3.5pt);
    \draw [fill=black] (23.,2.) circle (3.5pt);
    \draw[color=black] (22.677540386537533,6.9255253547141695) node {$G_j$};
    \draw [fill=black] (28.,6.) circle (3.5pt);
    \draw [fill=black] (29.,6.) circle (3.5pt);
    \draw [fill=black] (30.,6.) circle (3.5pt);
    \draw [fill=black] (31.,2.) circle (3.5pt);
    \draw [fill=black] (30.,2.) circle (3.5pt);
    \draw [fill=black] (29.,2.) circle (3.5pt);
    \draw[color=black] (28.353829852143367,6.787918337366149) node {$G_m$};
    \end{scriptsize}
    \end{tikzpicture}
    \caption{Illustration of the graph $G'_\phi$.}
    \label{fig:graph gphi path length three}
\end{figure}

\begin{lemma}
    \label{lem:G'max}
    No maximum matching of $G'_\phi$ contains an edge from $U'$.
\end{lemma}
\begin{proof}
    Let $M$ be a perfect matching of $G'_\phi, f'_\phi$ (see \cref{itm:pm} in \cref{sec:NP:atmosttwo}). First, observe that $M$ contains the edges $\{s_1, s_2\}, \{s_3, s_4\}$, so none of the edges of $U'$ that are incident to $S$ are in $M$.
    Consider a clause gadget $G_j$. If $M$ contains an odd number of edges from $U'$ connecting $G_j$ to some clause gadget(s) then $G_j$ must contain $M$-unmatched vertices since $G_j$ has an even number of vertices, so $M$ cannot be perfect. On the other hand, if $M$ contains exactly two edges from $U'$ connecting $G_j$ to two clause gadgets, then these edges are incident to $r_{j, 1}$ and $r_{j, 3}$ by the construction of $G'_\phi$. Therefore, the vertex $r_{j, 2}$ is $M$-unmatched, so $M$ cannot be perfect. It follows that $M$ contains no edge from $U'$.
\end{proof}

\subsubsection{Each color class is a paw}

We adapt the reduction from \cref{sec:NP:atmosttwo} in order to show the following.

\rmtriangleandedge*

To prove \cref{thm: rainbow matching triangle plus edge}, we first construct from a 3-occurrence 3SAT formula $\phi$ the graphs $G_\phi, f_\phi$ and $G'_\phi, f'_\phi$ as described in \cref{sec:NP:atmosttwo,sec:NP:atmostthree}. We then add to each color class of $G'_\phi, f_\phi$ that forms a $P_4$ one edge, such that we obtain a paw. We then show that a maximum matching of the resulting graph contains only edges of $G_\phi$ and invoke \cref{thm: Hilfslemma zu outerplanar rainbow matchings} to conclude.

We now describe the construction. Let $\phi$ be a 3-occurrence 3SAT formula on $n$ variables $x_1,\ldots, x_n$ and $m$ clauses $C_1, \ldots, C_m$. We construct, in polynomial time, the edge-colored graph $G'_\phi, f'_\phi$ as described in~\cref{sec:NP:atmostthree}.
For each clause gadget $C_j$, $1\leq j \leq m$, and each edge $\{r_{j,1},r'_{j,1}\}$, $\{r_{j,1},r'_{j,3}\}$ and $\{r_{j,3},r'_{j,3}\}$, we proceed as follows.
Let $\{a_i,u_i\}$ be the edge with same color as the $\{r_{j,1},r'_{j,1}\}$ according to $f'_\phi$. 
Then, by construction described in~\cref{thm: rainbow matching path length three}, there exists an edge of the same color connecting vertex $u_i$ to vertex $r_{j,1}$. 
In order to create a color class with the required properties, we add the edge $\{a_i,r_{j,1}\}$ of the same color to the graph. 
In general, we add the following edges to our graph.
\begin{itemize}
    \item If the edge $\{a_i, \overline{u}_i\}$ (resp., $\{a_i, u_i\}$) has the same color as either $\{r_{j,1}, r'_{j,1}\}$ or $\{r_{j,1}, r'_{j,3}\}$, then add the edge $\{a_i, r_{j,1}\}$ with the same color to the graph.
    
    \item If the edge $\{a_i, \overline{u}_i\}$ (resp., $\{a_i, u_i\}$) has the same color as $\{r_{j,3}, r'_{j,3}\}$, then add the edge $\{a_i, r_{j,3}\}$ with the same color to the graph.
    
    \item If the edge $\{\overline{u}_i, \overline{v}_i\}$ (resp., $\{u_i, v_i\}$) has the same color as either $\{r_{j,1}, r'_{j,1}\}$ or $\{r_{j,1}, r'_{j,3}\}$, then add the edge $\{\overline{u}_i, r_{j,1}\}$ (resp., $\{u_i, r_{j,1}\}$) with the same color to the graph.
    
    \item If the edge $\{\overline{u}_i, \overline{v}_i\}$ (resp., $\{u_i, v_i\}$) has the same color as $\{r_{j,3}, r'_{j,3}\}$, then add the edge $\{\overline{u}_i, r_{j,3}\}$ (resp., $\{u_i, r_{j,3}\}$) with the same color to the graph.
    
    \item If the edge $\{\overline{v}_i, b_i\}$ (resp., $\{v_i, b_i\}$) has the same color as either $\{r_{j,1}, r'_{j,1}\}$ or $\{r_{j,1}, r'_{j,3}\}$, then add the edge $\{\overline{v}_i, r_{j,1}\}$ (resp., $\{v_i, r_{j,1}\}$) with the same color to the graph.
    
    \item If the edge $\{\overline{v}_i, b_i\}$ (resp., $\{v_i, b_i\}$) has the same color as $\{r_{j,3}, r'_{j,3}\}$, then add the edge $\{\overline{v}_i, r_{j,3}\}$ (resp., $\{v_i, r_{j,3}\}$) with the same color to the graph.

    \item For each of the two color classes $\{\{s_1, s_2\}, \{s_3, s_4\}\}$, we add the edges $\{s_1, s_3\}, \{s_2, s_3\}$ of the same color.

    \item For each color class of the form $\{u, v\}, \{s_1, s_3\}$, where $u, v \in V(G_\phi)$, we add the edges $\{v, s_1\}, \{v, s_3\}$ of the same color.
\end{itemize}

Let $G''_\phi, f''_\phi$ be the resulting edge-colored graph. 
By the construction above, every color class is a paw. Let $U'' = E(G''_\phi) \setminus E(G'_\phi)$ be the edges added to $G'_\phi$ in the above construction. 
The proof of the following lemma is obtained by substituting $G'_\phi$ by $G''_\phi$ and $U'$ by $U''$ in the proof of \cref{lem:G'max}.

\begin{lemma}
    \label{lem:G''max}
    No maximum matching of $G''_\phi$ contains an edge from $U''$.
\end{lemma}

From \cref{lem:G''max,lem:G'max} we conclude that any maximum matching of $G''_\phi$ contains only edges from $G_\phi \subseteq G''_\phi$. \cref{thm: rainbow matching triangle plus edge} then follows by invoking \cref{thm: Hilfslemma zu outerplanar rainbow matchings}.

\subsection{Hardness-preserving structural conditions for Rainbow Matchings}
\label{sec:hardness}

In this subsection, we present a hardness-preserving structural property of graphs that is used in the proof of the dichotomy theorem.

\rmhardness*

\begin{proof}
    Fix some graph class $\mathcal{G}$ as in the lemma statement.
    By the definition of $\mathcal{G}$, there is some $C \in \{2K_2, P_4, \text{paw}\}$, such that for any $\ell' > 0$, there is some $G \in \mathcal{G}$, such that at least $\ell'$ color classes of $G$ contain $C$ as induced subgraph. Assume without loss of generality that the color classes $G_1, G_2, \ldots, G_{\ell'}$ are isomorphic to $C$ and let $G_{\ell'+1}, \ldots, G_q$, $q \geq \ell$ be the remaining non-trivial color classes of $G$.
    Furthermore, assume that $2K_2$ is an induced subgraph of $G_1, ..., G_{\ell'}$, i.e., $C$ is a $2K_2$.
    We then show the \NP-hardness for \RM on $\mathcal{G}$ by a polynomial-time reduction from \RM where every color class is a $2K_2$, which is \NP-hard by \cref{thm: rainbow matching NP-complete outerplanar}.
    The other two cases, in which $C$ is either a $P_4$ or a paw are analogous, and the \NP-hardness follows by either \cref{thm: rainbow matching path length three} or \cref{thm: rainbow matching triangle plus edge}.
    Hence, we only provide the proof for the case where $C$ is a $2K_2$.
    Note that since $\mathcal{G}$ is color-closed, every edge-colored graph $G'$ with non-trivial color classes isomorphic to $G_1, ..., G_{q}$ and arbitrarily many trivial color classes is contained in~$\mathcal{G}$.

    Let $H = (V, E; c)$ be an edge-colored graph with $\ell'$ color classes $H_1, ..., H_{\ell'}$, such that each color class induces a $2K_2$.
    We construct in polynomial time an instance $G'$ of \RM, such that the graph $G'$ is in $\mathcal{G}$, and a maximum rainbow matching of $H$ can be transformed in polynomial time into a maximum rainbow matching of $G'$ and vice versa. 
    This shows the \NP-hardness of \RM on~$\mathcal{G}$ by \cref{thm: rainbow matching NP-complete outerplanar}. 

    We proceed as follows. Let $H'$ be a copy of $H$.
    Then add to $H'$ a copy $H'_i$ of $G_i$ for $1 \leq i \leq q$.
    Note that the vertices of $H_i'$ are disjoint to $H_j'$ for $1 \leq i \leq q$.
    Now, for $1 \leq i \leq \ell'$, let $V'_i$ (resp., $W_i$) be four vertices of $H'_i$ (resp., $H_i$) that induce a $2K_2$. 
    For $1 \leq i \leq \ell'$, we then identify the vertices $V'_i$ and $W_i$ in a way that identifies the two $2K_2$s. 
    Note that the color classes of $H'$ are exactly $H_1', ..., H'_q$.
    Hence, the resulting graph $H'$ is in $\mathcal{G}$ and note that the vertices of $H_i'$ are disjoint from those of $H_j'$ whenever $\ell' +1 \leq i \leq q$.
    For an illustration of the graph $H'$, see \Cref{fig: matching underlying graph2}.
    The graph $H$ circled in gray contains $\ell' = q = 2$ color classes that have either black or blue edges and are glued to the copies of $G_i$ on the blue and the black $2K_2$.

    In a final step, for $1 \leq i \leq q$, we add to each vertex of $H'_i$ that has not been identified in the previous step a new edge to a new leaf vertex, which receives a new color; see \Cref{fig: matching underlying graph}.
    Let $X$ be the set of the edges added in this step. This concludes the construction of the graph $H'$, which is in~$\mathcal{G}$.
    \begin{figure}[t]
        \centering
        \begin{subfigure}{0.4\textwidth}
    \centering
    \definecolor{ududff}{rgb}{0.30196078431372547,0.30196078431372547,1.}
    \definecolor{yqyqyq}{rgb}{0.5019607843137255,0.5019607843137255,0.5019607843137255}
    \begin{tikzpicture}[line cap=round,line join=round,>=triangle 45,x=1.0cm,y=1.0cm, scale=0.5]
    \tikzstyle{every node}=[font=\large]
    
    \draw [line width=2.pt,color=ududff] (6.02,-2.48)-- (1.52,-2.12);
    \draw [line width=2.pt,color=ududff] (1.52,-2.12)-- (1.4,-0.54);
    \draw [line width=2.pt,color=ududff] (1.76,0.94)-- (1.3,2.08);
    \draw [line width=2.pt,color=ududff] (4.76,1.42)-- (1.3,2.08);
    \draw [line width=2.pt,color=ududff] (4.76,1.42)-- (1.4,-0.54);
    \draw [line width=2.pt,color=ududff] (1.4,-0.54)-- (1.76,0.94);
    \draw [line width=2.pt,color=ududff] (6.02,-2.48)-- (4.78,-0.84);
    \draw [line width=2.pt,color=ududff] (4.78,-0.84)-- (1.4,-0.54);
    \draw [line width=2.pt,color=ududff] (4.78,-0.02)-- (4.76,1.42);
    \draw [rotate around={-88.6452847489691:(5.87,-0.42)},line width=2.pt,color=yqyqyq] (5.87,-0.42) ellipse (3.6702787960688728cm and 2.168973591557258cm);
    \draw[color=yqyqyq] (5.98,2.51) node {$H$};
    
    \draw [line width=2.pt] (6.96,1.48)-- (8.98,-0.36);
    \draw [line width=2.pt] (9.52,2.14)-- (9.86,0.5);
    \draw [line width=2.pt] (8.98,-0.36)-- (9.54,-2.18);
    \draw [line width=2.pt] (9.54,-2.18)-- (9.86,0.5);
    \draw [line width=2.pt] (8.98,-0.36)-- (9.86,0.5);
    \draw [line width=2.pt] (6.96,1.48)-- (7.,0.);
    \draw [line width=2.pt] (7.,-0.86)-- (6.02,-2.48);
    \draw [line width=2.pt] (7.,-0.86)-- (8.98,-0.36);
    \draw [line width=2.pt] (6.96,1.48)-- (9.52,2.14);
    
    \begin{scriptsize}
    \draw [fill=ududff] (6.02,-2.48) circle (3.5pt);
    \draw [fill=ududff] (4.76,1.42) circle (3.5pt);
    \draw [fill=ududff] (1.52,-2.12) circle (3.5pt);
    \draw [fill=ududff] (1.4,-0.54) circle (3.5pt);
    \draw [fill=ududff] (1.76,0.94) circle (3.5pt);
    \draw [fill=ududff] (1.3,2.08) circle (3.5pt);
    \draw [fill=ududff] (4.78,-0.84) circle (3.5pt);
    \draw [fill=ududff] (4.78,-0.02) circle (3.5pt);
    \draw[color=ududff] (5.0,-1.7) node {};
    \draw[color=ududff] (5.19,0.94) node {};
    
    \draw [fill=black] (6.96,1.48) circle (3.5pt);
    \draw [fill=black] (6.02,-2.48) circle (3.5pt); %
    \draw [fill=black] (8.98,-0.36) circle (3.5pt);
    \draw [fill=black] (9.52,2.14) circle (3.5pt);
    \draw [fill=black] (9.86,0.5) circle (3.5pt);
    \draw [fill=black] (9.54,-2.18) circle (3.5pt);
    \draw [fill=black] (7.,0.) circle (3.5pt);
    \draw[color=black] (6.67,0.9) node {};
    \draw [fill=black] (7.,-0.86) circle (3.5pt);
    \draw[color=black] (6.93,-1.7) node {};
    \end{scriptsize}
    \end{tikzpicture}
    \caption{Result of the first step of the construction of the graph $H'$.}
    \label{fig: matching underlying graph2}
\end{subfigure}
        \begin{subfigure}{0.45\textwidth}
    \centering
    \definecolor{yqyqyq}{rgb}{0.5019607843137255,0.5019607843137255,0.5019607843137255}
    \definecolor{qqqqff}{rgb}{0.,0.,1.}
    \definecolor{qqwwtt}{rgb}{0.,0.4,0.2}
    \definecolor{yqqqqq}{rgb}{0.5019607843137255,0.,0.}
    \definecolor{ccwwff}{rgb}{0.8,0.4,1.}
    \definecolor{ffqqff}{rgb}{1.,0.,1.}
    \definecolor{xfqqff}{rgb}{0.4980392156862745,0.,1.}
    \definecolor{qqffqq}{rgb}{0.,1.,0.}
    \definecolor{ffxfqq}{rgb}{1.,0.4980392156862745,0.}
    \definecolor{ffqqqq}{rgb}{1.,0.,0.}
    \definecolor{aqaqaq}{rgb}{0.6274509803921569,0.6274509803921569,0.6274509803921569}
    \definecolor{ududff}{rgb}{0.30196078431372547,0.30196078431372547,1.}
    \begin{tikzpicture}[line cap=round,line join=round,>=triangle 45,x=1.0cm,y=1.0cm, scale=0.5]
    \tikzstyle{every node}=[font=\large]
    \draw [line width=2.pt,color=ududff] (5.94,-2.48)-- (1.52,-2.12);
    \draw [line width=2.pt,color=ududff] (1.52,-2.12)-- (1.4,-0.54);
    \draw [line width=2.pt,color=ududff] (1.76,0.94)-- (1.3,2.08);
    \draw [line width=2.pt,color=ududff] (4.76,1.42)-- (1.3,2.08);
    \draw [line width=2.pt,color=ududff] (4.76,1.42)-- (1.4,-0.54);
    \draw [line width=2.pt,color=ududff] (1.4,-0.54)-- (1.76,0.94);
    \draw [line width=2.pt] (6.96,1.48)-- (8.98,-0.36);
    \draw [line width=2.pt] (9.52,2.14)-- (9.86,0.5);
    \draw [line width=2.pt] (8.98,-0.36)-- (9.54,-2.18);
    \draw [line width=2.pt] (9.54,-2.18)-- (9.86,0.5);
    \draw [line width=2.pt] (8.98,-0.36)-- (9.86,0.5);
    \draw [line width=2.pt,color=ffqqqq] (1.52,-2.12)-- (-0.6,-1.98);
    \draw [line width=2.pt,color=ffxfqq] (-0.4,-0.18)-- (1.4,-0.54);
    \draw [line width=2.pt,color=qqffqq] (-0.12,1.62)-- (1.76,0.94);
    \draw [line width=2.pt,color=xfqqff] (1.3,2.08)-- (2.12,3.28);
    \draw [line width=2.pt,color=ffqqff] (10.86,2.78)-- (9.52,2.14);
    \draw [line width=2.pt,color=ccwwff] (11.68,0.72)-- (9.86,0.5);
    \draw [line width=2.pt,color=yqqqqq] (10.76,-2.76)-- (9.54,-2.18);
    \draw [line width=2.pt,color=qqwwtt] (11.26,-0.76)-- (8.98,-0.36);
    \draw [line width=2.pt,color=ududff] (5.94,-2.48)-- (4.78,-0.84);
    \draw [line width=2.pt,color=ududff] (4.78,-0.02)-- (4.76,1.42);
    \draw [line width=2.pt] (6.96,1.48)-- (7.,0.);
    \draw [line width=2.pt,color=ududff] (4.78,-0.84)-- (1.4,-0.54);
    \draw [line width=2.pt] (7.,-0.86)-- (8.98,-0.36);
    \draw [line width=2.pt] (6.96,1.48)-- (9.52,2.14);
    \draw [rotate around={-87.00153286030528:(1.41,-0.02)},line width=2.pt,dash pattern=on 5pt off 5pt,color=qqqqff] (1.41,-0.02) ellipse (2.3617812209291156cm and 1.0751327990222517cm);
    \draw [rotate around={-89.73474365662463:(9.53,-0.02)},line width=2.pt,dash pattern=on 5pt off 5pt] (9.53,-0.02) ellipse (2.417793084871198cm and 1.0862888203655128cm);
    \draw [rotate around={-88.6452847489691:(5.87,-0.42)},line width=2.pt,color=yqyqyq] (5.87,-0.42) ellipse (3.6702787960688728cm and 2.168973591557258cm);
    \draw [line width=2.pt] (7.,-0.86)-- (5.94,-2.48);
    \begin{scriptsize}
    \draw [fill=ududff] (5.94,-2.48) circle (3.5pt);
    \draw [fill=ududff] (4.76,1.42) circle (3.5pt);
    \draw [fill=black] (6.96,1.48) circle (3.5pt);
    \draw [fill=ududff] (1.52,-2.12) circle (3.5pt);
    \draw [fill=ududff] (1.4,-0.54) circle (3.5pt);
    \draw [fill=ududff] (1.76,0.94) circle (3.5pt);
    \draw [fill=ududff] (1.3,2.08) circle (3.5pt);
    \draw [fill=black] (8.98,-0.36) circle (3.5pt);
    \draw [fill=black] (9.52,2.14) circle (3.5pt);
    \draw [fill=black] (9.86,0.5) circle (3.5pt);
    \draw [fill=black] (9.54,-2.18) circle (3.5pt);
    \draw [fill=aqaqaq] (10.86,2.78) circle (2.5pt);
    \draw [fill=aqaqaq] (11.68,0.72) circle (2.5pt);
    \draw [fill=aqaqaq] (10.76,-2.76) circle (2.5pt);
    \draw [fill=aqaqaq] (11.26,-0.76) circle (2.5pt);
    \draw [fill=aqaqaq] (-0.6,-1.98) circle (2.5pt);
    \draw [fill=aqaqaq] (2.12,3.28) circle (2.5pt);
    \draw [fill=aqaqaq] (-0.12,1.62) circle (2.5pt);
    \draw [fill=aqaqaq] (-0.4,-0.18) circle (2.5pt);
    \draw [fill=ududff] (4.78,-0.84) circle (3.5pt);
    \draw[color=ududff] (5.0,-1.7) node { };
    \draw [fill=ududff] (4.78,-0.02) circle (3.5pt);
    \draw[color=ududff] (5.167510367371316,0.9166442597682075) node {};
    \draw [fill=black] (7.,0.) circle (3.5pt);
    \draw[color=black] (6.668818066142036,0.8976403648470592) node {};
    \draw [fill=black] (7.,-0.86) circle (3.5pt);
    \draw[color=yqyqyq] (6.070195376125863,2.4654616958417965) node {$H$};
    \draw[color=black] (6.93,-1.7) node {};
    \end{scriptsize}
    \end{tikzpicture}
    \caption{Result of the second step of the construction of the graph $H'$.}
    \label{fig: matching underlying graph}
\end{subfigure}
        \caption{Illustration of the two steps in the construction of the graph $H'$ in \cref{lem:hardness-part-for-dichotomy}.\label{fig:placeholder}}
    \end{figure}
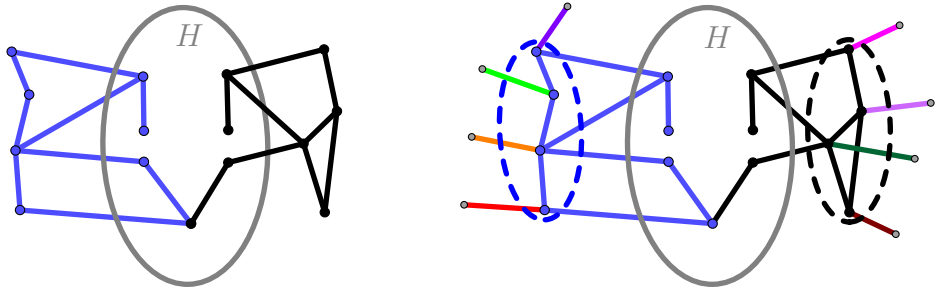

    We now show that $H$ has a maximum rainbow matching of size $p$ if and only if $H'$ has a maximum rainbow matching of size 
    $p' = p - 4\ell + \sum_{1 \leq i \leq r} |V(H_i')|$.
    Let $M$ be a maximum rainbow matching of size $p$ in $H$.
    To construct a rainbow matching of size $p'$, observe that $M \cup X$ is also a rainbow matching and that $|M \cup X| = p'$.
    Now let $M'$ be a maximum rainbow matching of size $p'$ in $H'$.
    Note that there exists a maximum rainbow matching $M''$ of size $p'$ in $G'$ such that $X \subseteq M''$, since every edge of $X$ is incident to a leaf vertex and its color class is of size $1$.
    Now observe that $M = M'' \setminus X$ is also a rainbow matching in $H$ and note that $|M| = p$.
\end{proof}

\subsection{Polynomial-time cases: CM-colored graphs}
\label{sec:polynomial:mm} 

In this section, we prove the polynomial part of \cref{thm: dichotomy result rainbow matchings}.
For this purpose, we consider the more general problem \textsc{Maximum Generalized Rainbow Matching}, where we are allowed to select at most $m_i \geq 1$ edges from each color class $i$. 
We first show that finding such a matching of maximum cardinality admits a polynomial-time algorithm when the instances are restricted to strictly CM-colored graphs (\cref{thm:dichotomy-result-rainbow-l-u-matchings}). 
We prove this by a reduction to \DCSshort, which is a refinement of the one presented in \cite[Section 7]{matching-complete-graph}.
Finally, to obtain \Cref{thm: dichotomy result rainbow matchings poly cases}, the positive part of the dichotomy \cref{thm: dichotomy result rainbow matchings}, we consider $\alpha$-CM-colored graphs and perform a simple enumeration of rainbow matchings on the at most $\alpha$ color classes that are not CM and apply \cref{thm:dichotomy-result-rainbow-l-u-matchings}, letting $m_i = 1$ for the remaining color classes.

\theoremmaintwo*

Notice that the values $m_i$ are part of the input.
In order to obtain \cref{thm:dichotomy-result-rainbow-l-u-matchings}, we cannot directly apply the results from~\cite{matching-complete-graph}, since they rely on the non-trivial color classes being almost vertex-disjoint: If they are not, the constructed matching is not guaranteed to be maximum. 
Furthermore, in~\cite{matching-complete-graph}, each color class is a clique.
We call a matching $M$ of $G$ \emph{$m_i$-restricted} if $|M \cap C_i| \leq m_i$ for each color class $i$, i.e., $M$ has at most $m_i$ edges from color class $i$.
To prove~\cref{thm:dichotomy-result-rainbow-l-u-matchings}, from the input given in~\cref{thm:dichotomy-result-rainbow-l-u-matchings} we construct an equivalent instance $(G', l', u')$ of \DCSshort and obtain a maximum $(l', u')$-matching of $G'$ using the polynomial-time algorithm from~\cite{Gabow,Shiloach}.
By establishing a correspondence between $(l', u')$-matchings of $(G', l', u')$ and $m_i$-restricted 
matchings of $G$, we finally obtain~\cref{thm:dichotomy-result-rainbow-l-u-matchings}.

First, we show that we can assume $m_i=1$ for each color class $i$:
For each color $i$ with $m_i>1$,  
we create $m_i$ copies of the edges of color class $C_i$ to obtain color classes $C_i^1, ..., C_i^{m_i}$ and allow at most $1$ edge to be selected from each color class $C_i^j$, $j \in \{1, ..., m_i\}$. 
Observe that a solution to one instance can be turned into a solution to the other instance such that the solution sizes are the same. 
This is true since in the new instance one can also select up to $m_i$ edges from the copies of color class~$i$. 
Furthermore, since each vertex is incident to at most one edge of the solution, no additional copies can be selected compared to the original graph.
Hence, from now on we assume $m_i=1$ for each color class $i$.

We now describe the construction of the instance $(G', l', u')$ of \DCS.
We first take $G$ and set $l'(v) = 0$ and $u'(v) = 1$ for all $v \in V(G)$, i.e., for each original vertex $v \in V(G)$ we set its lower bound to $0$ and its upper bound to $1$.
Let $C_1, C_2, \ldots, C_k$ be the color classes of $G$ and for $1 \leq i \leq k$, let $n_i = |E(C_i)|$. 
The idea is to construct a gadget for each color class $C_i$.
For each $1 \leq i \leq k$, we proceed as follows.
We replace each edge $\{p,q\}$ of $C_i$ with three new edges $\{p,v^p_q\},\{v^p_q,v^q_p\},\{v^q_p,q\}$. 
The two edges $\{p,v^p_q\}$ and $\{v^q_p,q\}$ are called \emph{half-edges} of $\{p,q\}$ and the edge $\{v^p_q,v^q_p\}$ is called the \emph{eliminator} of $\{p,q\}$. The vertices $v^p_q, v^q_p$ will be called \emph{subdivision vertices} and have a degree interval of $[1,1]$. 

For each color class $1\leq i \leq k$, let $\mathcal{C}_i=\{K_1^i,\ldots, K_{\ell_i}^i\}$ be the set of non-trivial parts from the corresponding complete multipartite graph. 
Notice that $\mathcal{C}_i$ can be computed in polynomial time for each color class $i$, since for the complete multipartite graph $C_i$, the parts $K^i_1, \ldots K_{\ell_i}^i$ are vertex-disjoint.
We introduce $\ell_i$ new \emph{local} vertices  $u_1^i, u_2^i, \ldots , u_{\ell_i}^i$, one for each non-trivial part.  
For each non-trivial part $K_j^i$ in $\mathcal{C}_i$, we connect for each vertex $q$ of $K^i_j$ each subdivision vertex of $C_i$ adjacent to $q$ to $u^i_j$.
The degree interval of a local vertex $u_j^i, 1\leq j \leq \ell_i$ is $[\deg(u_j^i)-1,\deg(u_j^i)-1]$.
The local vertices will later ensure that we pick at most one half-edge from each non-trivial part $K_j^i$.

Next, we introduce a \emph{universal} vertex $u^i_{\ell_i+1}$ for $1 \leq i \leq k$. 
Each subdivision vertex obtained from an edge of $C_i$ is connected to the universal vertex $u_{\ell_i+1}^i$.
The degree interval of this universal vertex is $[2 n_i -\sum_{j=1}^{l_i} u'(u_j^i) -2, 2 n_i-\sum_{j=1}^{l_i} u'(u_j^i) -2]$, where $u'(u_j^i)$ is the upper degree bound of the degree interval of $u_j^i$ from above, which is $\deg(u_j^i)-1$. 
Note that $ 2 n_i - \sum_{j=1}^{l_i} u'(u_j^i) - 2 \geq 0$.
This universal vertex ensures that any subgraph respecting the degree bounds contains at most $2$ half-edges from each color class. 
This concludes the construction. 

Please refer to \Cref{fig: example thm 4.5 not transformed graph,fig: example thm 4.5 graph G'} for an example of the construction of $G'$ from a graph~$G$, whose color classes are two $K_3$s, a $K_{1, 4}$ and a $K_{2, 2}$. 
Notice that each color class in the example is complete multipartite.
The red and green color classes have one edge in common. Similarly, the green and purple color classes have one edge in common. The other color classes do not have edges in common. The new vertices added in the construction of $G'$ are colored according to their color class. The other vertices are black. The local vertices have the shape of a triangle and the universal vertices the shape of a rhombus. 

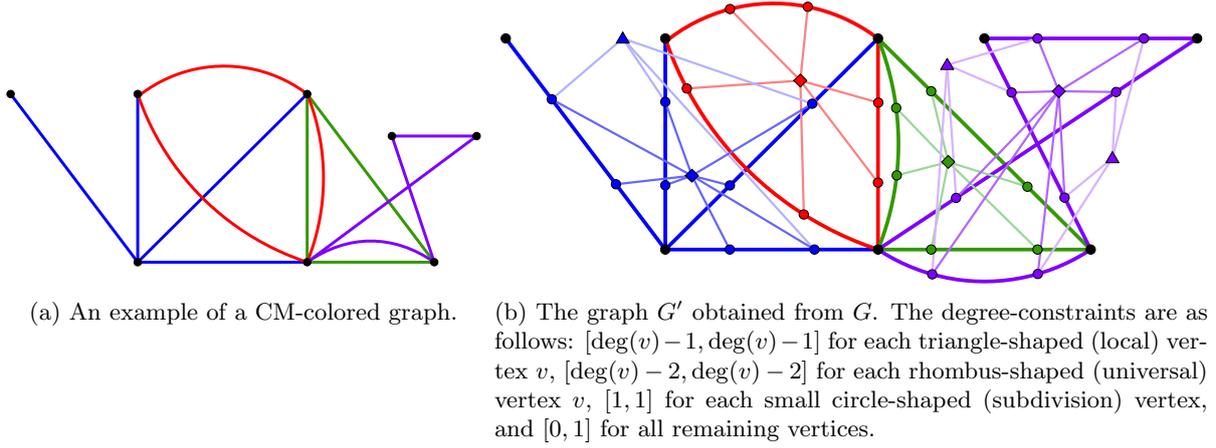
\begin{figure}  
    \centering
    \begin{subfigure}[t]{0.4\textwidth}
    \centering
    \definecolor{xfqqff}{rgb}{0.4980392156862745,0.,1.}
    \definecolor{xfqqffLight}{rgb}{0.7,0.4,1.}
    \definecolor{xfqqffVLight}{rgb}{0.85,0.7,1.}
    \definecolor{qqqqff}{rgb}{0.,0.,1.}
    \definecolor{qqqqffLight}{rgb}{0.4,0.4,1.}
    \definecolor{qqqqffVLight}{rgb}{0.7,0.7,1.}
    \definecolor{ffqqqq}{rgb}{1.,0.,0.}
    \definecolor{ffqqqqLight}{rgb}{1.,0.5,0.5}
    \definecolor{ttzzqq}{rgb}{0.2,0.6,0.}
    \definecolor{ttzzqqLight}{rgb}{0.6,0.85,0.6}
    \resizebox{\textwidth}{!}{
    \begin{tikzpicture}[line cap=round,line join=round,>=triangle 45,x=1.0cm,y=1.0cm]
    \draw [line width=2.pt,color=ttzzqq] (5.,4.)-- (8.,0.);
    \draw [line width=2.pt,color=ttzzqq] (8.,0.)-- (5.,0.);
    \draw [line width=2.pt,color=ttzzqq] (5.,0.)-- (5.,4.);
    \draw [shift={(0.,2.)},line width=2.pt,color=ffqqqq]  plot[domain=-0.3805063771123649:0.3805063771123649,variable=\t]({1.*5.385164807134505*cos(\t r)+0.*5.385164807134505*sin(\t r)},{0.*5.385164807134505*cos(\t r)+1.*5.385164807134505*sin(\t r)});
    \draw [line width=2.pt,color=qqqqff] (1.,0.)-- (1.,4.);
    \draw [line width=2.pt,color=qqqqff] (5.,0.)-- (1.,0.);
    \draw [line width=2.pt,color=qqqqff] (1.,0.)-- (5.,4.);
    \draw [shift={(7.,6.)},line width=2.pt,color=ffqqqq]  plot[domain=3.4633432079864352:4.3906384259880475,variable=\t]({1.*6.324555320336759*cos(\t r)+0.*6.324555320336759*sin(\t r)},{0.*6.324555320336759*cos(\t r)+1.*6.324555320336759*sin(\t r)});
    \draw [line width=2.pt,color=xfqqff] (9.,3.)-- (7.,3.);
    \draw [line width=2.pt,color=xfqqff] (7.,3.)-- (8.,0.);
    \draw [line width=2.pt,color=xfqqff] (9.,3.)-- (5.,0.);
    \draw[line width=2.pt,color=xfqqff] (5,0) to[out=35,in=145] (8,0);
    \draw [shift={(3.04,1.44)},line width=2.pt,color=ffqqqq]  plot[domain=0.9173699856141347:2.24362887438866,variable=\t]({1.*3.2241588050218617*cos(\t r)+0.*3.2241588050218617*sin(\t r)},{0.*3.2241588050218617*cos(\t r)+1.*3.2241588050218617*sin(\t r)});
    \draw [line width=2.pt,color=qqqqff] (1.,0.)-- (-2.,4.);
    \begin{scriptsize}
    \draw [fill=black] (5.,4.) circle (2.5pt);
    \draw[color=black] (4.96,4.41) node {};
    \draw [fill=black] (5.,0.) circle (2.5pt);
    \draw[color=black] (4.96,-0.41) node {};
    \draw [fill=black] (8.,0.) circle (2.5pt);
    \draw [fill=black] (7.,3.) circle (2.5pt);
    \draw [fill=black] (1.,0.) circle (2.5pt);
    \draw [fill=black] (1.,4.) circle (2.5pt);
    \draw [fill=black] (9.,3.) circle (2.5pt);
    \draw [fill=black] (-2.,4.) circle (2.5pt);
    \end{scriptsize}
    \end{tikzpicture}
    }
    \caption{An example of a CM-colored graph.\label{fig: example thm 4.5 not transformed graph}}
\end{subfigure}
\begin{subfigure}[t]{0.59\textwidth}
    \centering
    \definecolor{ududff}{rgb}{0.30196078431372547,0.30196078431372547,1.}
    \definecolor{xfqqff}{rgb}{0.4980392156862745,0.,1.}
    \definecolor{xfqqffLight}{rgb}{0.7,0.4,1.}
    \definecolor{xfqqffVLight}{rgb}{0.85,0.7,1.}
    \definecolor{qqqqff}{rgb}{0.,0.,1.}
    \definecolor{qqqqffLight}{rgb}{0.4,0.4,1.}
    \definecolor{qqqqffVLight}{rgb}{0.7,0.7,1.}
    \definecolor{ffqqqq}{rgb}{1.,0.,0.}
    \definecolor{ffqqqqLight}{rgb}{1.,0.5,0.5}
    \definecolor{ttzzqq}{rgb}{0.2,0.6,0.}
    \definecolor{ttzzqqLight}{rgb}{0.6,0.85,0.6}

    \resizebox{\textwidth}{!}{
    \begin{tikzpicture}[line cap=round,line join=round,>=triangle 45,x=1.0cm,y=1.0cm]
    \draw [line width=2.pt,color=ttzzqq] (5.,4.)-- (9.,0.);
    \draw [line width=2.pt,color=ttzzqq] (9.,0.)-- (5.,0.);
    \draw [line width=2.pt,color=ffqqqq] (5.,0.)-- (5.,4.);
    \draw [shift={(0.,2.)},line width=2.pt,color=ttzzqq]  plot[domain=-0.3805063771123649:0.3805063771123649,variable=\t]({1.*5.385164807134505*cos(\t r)+0.*5.385164807134505*sin(\t r)},{0.*5.385164807134505*cos(\t r)+1.*5.385164807134505*sin(\t r)});
    \draw [line width=2.pt,color=qqqqff] (1.,0.)-- (1.,4.);
    \draw [line width=2.pt,color=qqqqff] (5.,0.)-- (1.,0.);
    \draw [line width=2.pt,color=qqqqff] (1.,0.)-- (5.,4.);
    \draw [shift={(7.,6.)},line width=2.pt,color=ffqqqq]  plot[domain=3.4633432079864352:4.3906384259880475,variable=\t]({1.*6.324555320336759*cos(\t r)+0.*6.324555320336759*sin(\t r)},{0.*6.324555320336759*cos(\t r)+1.*6.324555320336759*sin(\t r)});
    \draw [line width=2.pt,color=xfqqff] (11.,4.)-- (7.,4.);
    \draw [line width=2.pt,color=xfqqff] (7.,4.)-- (9.,0.);
    \draw [line width=2.pt,color=xfqqff] (11.,4.)-- (5.,0.);
    \draw [shift={(3.04,1.44)},line width=2.pt,color=ffqqqq]  plot[domain=0.9173699856141347:2.24362887438866,variable=\t]({1.*3.2241588050218617*cos(\t r)+0.*3.2241588050218617*sin(\t r)},{0.*3.2241588050218617*cos(\t r)+1.*3.2241588050218617*sin(\t r)});
    \draw [shift={(7.,3.)},line width=2.pt,color=xfqqff]  plot[domain=4.124386376837122:5.3003915839322575,variable=\t]({1.*3.6055512754639896*cos(\t r)+0.*3.6055512754639896*sin(\t r)},{0.*3.6055512754639896*cos(\t r)+1.*3.6055512754639896*sin(\t r)});
    \draw [line width=1.2pt,color=ffqqqqLight] (1.4045914103554615,3.0518814957786855)-- (3.5394572742335875,3.2021275509821447);
    \draw [line width=1.2pt,color=ffqqqqLight] (3.5394572742335875,3.2021275509821447)-- (2.216578462580562,4.557238677373583);
    \draw [line width=1.2pt,color=ffqqqqLight] (3.5394572742335875,3.2021275509821447)-- (3.68,4.6);
    \draw [line width=1.2pt,color=ffqqqqLight] (3.6045012494913373,0.6642162492006749)-- (3.5394572742335875,3.2021275509821447);
    \draw [line width=1.2pt,color=ffqqqqLight] (5.,1.2680391818196652)-- (3.5394572742335875,3.2021275509821447);
    \draw [line width=1.2pt,color=ffqqqqLight] (5.005824926670503,2.784971236064747)-- (3.5394572742335875,3.2021275509821447);
    \draw [line width=1.2pt,color=ttzzqqLight] (6.320499373682906,1.6599132958329788)-- (7.809978323958248,1.1900216760417521);
    \draw [line width=1.2pt,color=ttzzqqLight] (8.,0.)-- (6.320499373682906,1.6599132958329788);
    \draw [line width=1.2pt,color=ttzzqqLight] (6.,0.)-- (6.320499373682906,1.6599132958329788);
    \draw [line width=1.2pt,color=ttzzqqLight] (5.352207091580496,1.405122492578971)-- (6.320499373682906,1.6599132958329788);
    \draw [line width=1.2pt,color=ttzzqqLight] (5.341658624505617,2.6831421076502076)-- (6.320499373682906,1.6599132958329788);
    \draw [line width=1.2pt,color=ttzzqqLight] (6.,3.)-- (6.320499373682906,1.6599132958329788);
    \draw [line width=1.2pt,color=xfqqffVLight] (9.40783034448357,1.7328054621766231)-- (10.,4.);
    \draw [line width=1.2pt,color=xfqqffVLight] (9.477878472205377,2.985252314803585)-- (9.40783034448357,1.7328054621766231);
    \draw [line width=1.2pt,color=xfqqffVLight] (9.40783034448357,1.7328054621766231)-- (8.005096801277759,-0.462626231642858);
    \draw [line width=1.2pt,color=xfqqffVLight] (9.40783034448357,1.7328054621766231)-- (8.508036181717761,0.9839276365644767);
    \draw [line width=2.pt,color=qqqqff] (1.,0.)-- (-2.,4.);
    \draw [line width=1.2pt,color=xfqqffVLight] (6.3,3.5)-- (8.,4.);
    \draw [line width=1.2pt,color=xfqqffVLight] (6.3,3.5)-- (7.511917466096825,2.9761650678063503);
    \draw [line width=1.2pt,color=xfqqffVLight] (6.3,3.5)-- (6.467351779934366,0.9782345199562439);
    \draw [line width=1.2pt,color=xfqqffVLight] (6.3,3.5)-- (6.019059748455469,-0.4695469766094469);
    \draw [line width=1.2pt,color=qqqqffLight] (0.07228005134939886,1.2369599315341349)-- (1.5,1.4);
    \draw [line width=1.2pt,color=qqqqffLight] (1.,1.2174747800114967)-- (1.5,1.4);
    \draw [line width=1.2pt,color=qqqqffLight] (2.2148082533903164,1.2148082533903164)-- (1.5,1.4);
    \draw [line width=1.2pt,color=qqqqffLight] (2.2121417267691363,0.)-- (1.5,1.4);
    \draw [line width=1.2pt,color=qqqqffVLight] (-1.1351978638296858,2.846930485106248)-- (0.2,4);
    \draw [line width=1.2pt,color=qqqqffVLight] (1.,2.79761233651679)-- (0.2,4);
    \draw [line width=1.2pt,color=qqqqffVLight] (3.763343058765504,2.763343058765504)-- (0.2,4);
    \draw [line width=1.2pt,color=qqqqffVLight] (3.804920383726472,0.)-- (0.2,4);
    \draw [line width=1.2pt,color=qqqqffLight] (-1.1351978638296858,2.846930485106248)-- (1.5,1.4);
    \draw [line width=1.2pt,color=qqqqffLight] (1.,2.79761233651679)-- (1.5,1.4);
    \draw [line width=1.2pt,color=qqqqffLight] (3.763343058765504,2.763343058765504)-- (1.5,1.4);
    \draw [line width=1.2pt,color=qqqqffLight] (3.804920383726472,0.)-- (1.5,1.4);
    \draw [line width=1.2pt,color=xfqqffLight] (8.4,3)-- (10.,4.);
    \draw [line width=1.2pt,color=xfqqffLight] (9.477878472205377,2.985252314803585)-- (8.4,3);
    \draw [line width=1.2pt,color=xfqqffLight] (8.4,3)-- (8.005096801277759,-0.462626231642858);
    \draw [line width=1.2pt,color=xfqqffLight] (8.4,3)-- (8.508036181717761,0.9839276365644767);
    \draw [line width=1.2pt,color=xfqqffLight] (8.4,3)-- (8.,4.);
    \draw [line width=1.2pt,color=xfqqffLight] (8.4,3)-- (7.511917466096825,2.9761650678063503);
    \draw [line width=1.2pt,color=xfqqffLight] (8.4,3)-- (6.467351779934366,0.9782345199562439);
    \draw [line width=1.2pt,color=xfqqffLight] (8.4,3)-- (6.019059748455469,-0.4695469766094469);
    
    \begin{scriptsize}
    \draw [fill=black] (5.,4.) circle (2.5pt);
    \draw [fill=black] (5.,0.) circle (2.5pt);
    \draw [fill=black] (9.,0.) circle (2.5pt);
    \draw [fill=black] (7.,4.) circle (2.5pt);
    \draw [fill=black] (1.,0.) circle (2.5pt);
    \draw [fill=black] (1.,4.) circle (2.5pt);
    \draw [fill=black] (11.,4.) circle (2.5pt);
    \draw [fill=ffqqqq] (1.4045914103554615,3.0518814957786855) circle (2.5pt);
    \draw [fill=ffqqqq] (3.6045012494913373,0.6642162492006749) circle (2.5pt);
    \draw [fill=ffqqqq] (2.216578462580562,4.557238677373583) circle (2.5pt);
    \draw [fill=ffqqqq] (3.68,4.6) circle (2.5pt);
    \draw [fill=ffqqqq] (5.005824926670503,2.784971236064747) circle (2.5pt);
    \draw [fill=ffqqqq] (5.,1.2680391818196652) circle (2.5pt);
    \draw [fill=ffqqqq] (3.5394572742335875,3.2021275509821447) ++(-3.5pt,0 pt) -- ++(3.5pt,3.5pt)--++(3.5pt,-3.5pt)--++(-3.5pt,-3.5pt)--++(-3.5pt,3.5pt);
    \draw [fill=ttzzqq] (5.341658624505617,2.6831421076502076) circle (2.5pt);
    \draw [fill=ttzzqq] (5.352207091580496,1.405122492578971) circle (2.5pt);
    \draw [fill=ttzzqq] (6.,0.) circle (2.5pt);
    \draw [fill=ttzzqq] (8.,0.) circle (2.5pt);
    \draw [fill=ttzzqq] (7.809978323958248,1.1900216760417521) circle (2.5pt);
    \draw [fill=ttzzqq] (6.,3.) circle (2.5pt);
    \draw [fill=ttzzqq] (6.320499373682906,1.6599132958329788) ++(-3.5pt,0 pt) -- ++(3.5pt,3.5pt)--++(3.5pt,-3.5pt)--++(-3.5pt,-3.5pt)--++(-3.5pt,3.5pt);
    \draw [fill=xfqqff] (8.,4.) circle (2.5pt);
    \draw [fill=xfqqff] (10.,4.) circle (2.5pt);
    \draw [fill=xfqqff] (7.511917466096825,2.9761650678063503) circle (2.5pt);
    \draw [fill=xfqqff] (9.477878472205377,2.985252314803585) circle (2.5pt);
    \draw [fill=xfqqff] (6.467351779934366,0.9782345199562439) circle (2.5pt);
    \draw [fill=xfqqff] (6.019059748455469,-0.4695469766094469) circle (2.5pt);
    \draw [fill=xfqqff] (8.005096801277759,-0.462626231642858) circle (2.5pt);
    \draw [fill=xfqqff] (9.40783034448357,1.7328054621766231) ++(0,3.5pt)  -- ++(-3.5pt,-6pt)  -- ++(7pt,0)   -- cycle;
    \draw [fill=xfqqff] (8.508036181717761,0.9839276365644767) circle (2.5pt);
    \draw [fill=black] (-2.,4.) circle (2.5pt);
    \draw [fill=xfqqff] (6.3,3.5)++(0,3.5pt)  -- ++(-3.5pt,-6pt)  -- ++(7pt,0)   -- cycle;
    \draw [fill=xfqqff] (8.4,3) ++(-3.5pt,0 pt) -- ++(3.5pt,3.5pt)--++(3.5pt,-3.5pt)--++(-3.5pt,-3.5pt)--++(-3.5pt,3.5pt);
    \draw [fill=qqqqff] (1.,2.79761233651679) circle (2.5pt);
    \draw [fill=qqqqff] (1.,1.2174747800114967) circle (2.5pt);
    \draw [fill=qqqqff] (2.2121417267691363,0.) circle (2.5pt);
    \draw [fill=qqqqff] (3.804920383726472,0.) circle (2.5pt);
    \draw [fill=qqqqff] (2.2148082533903164,1.2148082533903164) circle (2.5pt);
    \draw [fill=qqqqff] (3.763343058765504,2.763343058765504) circle (2.5pt);
    \draw [fill=qqqqff] (-1.1351978638296858,2.846930485106248) circle (2.5pt);
    \draw [fill=qqqqff] (0.07228005134939886,1.2369599315341349) circle (2.5pt);
    \draw [fill=qqqqff] (1.5,1.4) ++(-3.5pt,0 pt) -- ++(3.5pt,3.5pt)--++(3.5pt,-3.5pt)--++(-3.5pt,-3.5pt)--++(-3.5pt,3.5pt);
    \draw [fill=qqqqff] (0.2,4) ++(0,3.5pt)  -- ++(-3.5pt,-6pt)  -- ++(7pt,0)   -- cycle;
    \end{scriptsize}
    \end{tikzpicture}    
    }
    \caption{The graph $G'$ obtained from $G$. The degree-constraints are as follows: $[\deg(v)-1, \deg(v)-1]$ for each triangle-shaped (local) vertex $v$, $[\deg(v) -2, \deg(v)-2]$ for each rhombus-shaped (universal) vertex $v$, $[1, 1]$ for each small circle-shaped (subdivision) vertex, and $[0, 1]$ for all remaining vertices.
    \label{fig: example thm 4.5 graph G'}}
\end{subfigure}

    \caption{Example of the construction used in the reduction of \cref{thm:dichotomy-result-rainbow-l-u-matchings}.\label{fig:DCS example}}
\end{figure}
The above construction has the following intuition. 
Our goal is that selecting an edge in $G$ should correspond to selecting two \emph{compatible} half-edges in $G'$ and vice versa.
To ensure that we can transform a rainbow matching of $G$ into an $(l, u)$-matching of $G'$ and vice versa, for each color class $i$ we introduced two kinds of new vertices in $G'$: local vertices for each non-trivial part and one universal vertex.
The upper and lower bounds on the local vertices for each non-trivial part $K_j^i$ ensure that at most one half-edge per non-trivial part is selected in a feasible $(l', u')$-matching of $G'$.
Hence, in $G$ we do not match in color class $i$ two vertices from the same non-trivial part $K_j^i$, as desired.
Furthermore, the universal vertex of $G'$ must be incident to  $2 n_i-\sum_{j=1}^{l_i} u'(u_j^i) -2$ edges in any feasible $(l', u')$-matching of $G'$. 
Together with the degree bounds of the local vertices, this ensures that at most $2$ half-edges of $C_i$ are selected.
This corresponds to selecting in $G$ at most one edge per color class, as desired.
Finally, the degree constraint of exactly $1$ for the subdivision vertices implies that either exactly 0 or exactly 2 half-edges are selected. 
Hence, either we can select $0$ eliminator edges in $G'$, which results in selecting exactly $2$ half-edges in $G'$ and thus one edge in $G$, or we select $1$ eliminator edge in $G'$, which results in selecting $0$ half-edges in $G'$ and thus $0$ edges in $G$. We formalize this intuition in the claims below.

In the following, let $M$ be a (not necessarily maximum) rainbow matching in $G$ and $M'$ be a feasible (not necessarily maximum) solution to the instance $(G', l', u')$ of \DCSshort from above.
We refer to half-edges and eliminator edges from color class $i$ to be those edges of $G'$ that resulted from subdividing the edges of color class $i$ from $G$.
Similarly, we refer to the subdivision, universal and local vertices from color class $i$ in $G'$ to be those vertices introduced for color class $i$ from $G$.
Note that the lower and upper degree bounds of subdivision, local, and universal vertices are equal. Hence, we simply refer to their lower or upper degree bounds as \emph{degree bounds} or \emph{degree constraints}.

\begin{claim}
    \label{claim:structure:DCS}
    Let $M'\subseteq E(G')$ be a feasible solution to the instance $(G', l', u')$ of \DCSshort from above. Then the following hold:
    \begin{enumerate}
        \item  For each color class $i$, the solution $M'$ either contains $0$ eliminator edges and $2$ half-edges of color class $i$, or $M'$ contains $1$ eliminator edge and $0$ half-edges of color class $i$.
        \item If $M'$ contains $2$ half-edges of color class $i$, then these half-edges do not belong to the same non-trivial part $K_j^i$ from color class $i$.
    \end{enumerate}
\end{claim}

\begin{proof}
    It suffices to consider the edges and vertices introduced in $G'$ for any color class $i$ of $G$. Hence, for readability we sometimes omit that we refer to color class $i$.
    
    We first prove the first statement.
    Note that the total sum of degree bounds of local vertices and the universal vertex is $2 n_i -2$, where $2 n_i$ is twice the number of edges of color class $i$, which corresponds to the number of subdivision vertices introduced for color class $i$.
    These subdivision vertices have a degree bound of $1$. 
    Hence, in $M'$ at most $2$ subdivision vertices are not matched to the universal and local vertices.
    In~$M'$, these subdivision vertices can only be incident to a half-edge or an eliminator edge.
    However, since an eliminator edge is incident to exactly $2$ subdivision vertices, in $M'$ these subdivision vertices must be either incident to $1$ eliminator edge and $0$ half-edges, or incident to $0$ eliminator edges and $2$ half-edges.
    This  proves the first statement.

    It remains to prove the second statement.
    Assume that $M'$ contains $2$ half-edges of color class $i$ and assume for the sake of contradiction that these edges belong to the same non-trivial part $K_j^i$ of color class $i$.
    By construction of $G'$, there is a local vertex $u_j^i$ introduced for clique $K_j^i$ that is incident to all subdivision vertices introduced for edges going out of the vertices of $K_j^i$. 
    Let this set of subdivision vertices be $S_j^i$.
    The degree interval of the local vertex $u_j^i$ is $[\deg(u_j^i)-1, \deg(u_j^i)-1]$, where $|S_j^i| = \deg(u_j^i)$, and therefore at most one vertex of $S_j^i$ is not matched to $u_j^i$ in $M'$. This contradicts our assumption, and hence completes the proof.
\end{proof}

We next show that a rainbow matching $M$ of $G$ corresponds to a feasible solution $M'$ to $G'$ and vice versa, such that twice the number of edges in $M$ corresponds to the number of half-edges in $M'$.
For this purpose, let $h(M')$ be the number of half-edges in $M'$.

\begin{claim} 
    \label{claim:DCS:quasiclique}
    A rainbow matching $M$ of $G$ can be transformed in polynomial time into a feasible solution $M'$ of $G'$ such that $2|M| = h(M')$ and vice versa.
\end{claim}

\begin{proof}
    For the first direction, let $M$ be a rainbow matching of $G$.
    To construct $M'$, for each edge $e \in M$ we select the corresponding half-edges in $G'$ in the color class of $e$.
    For a color class $i$ for which $M$ does not have an edge, we select an arbitrary eliminator edge to be in $M'$. Next, we can select a set of edges incident to the local vertices $u_j^i$ of class $i$ such that $\deg(u_j^i)-1$ many edges are incident to the local vertices $u_j^i$.
    Finally, there are exactly $2 n_i -\sum_{j=1}^{l_i} u'(u_j^i) -2$ many unmatched subdivision vertices, which we can match to the universal vertex $u_{j+1}^i$ of color class $i$, which is exactly the degree constraint of this universal vertex.
    Note that we have $h(M') = 2 |M|$ by construction.
    Furthermore, note that $M'$ is feasible since all subdivision, local, and universal vertices satisfy their degree constraints and that each original vertex of $G$ in $G'$ is incident to at most one edge.
    Hence, $M'$ is a feasible solution to $(G', l', u')$ with $h(M') = 2 |M|$.

    Next, let $M'$ be a feasible solution to $G'$. By the first statement of \Cref{claim:structure:DCS}, we have that for each color class $i$, the solution $M'$ either contains $0$ eliminator edges and $2$ half-edges of color class $i$, or $M'$ contains $1$ eliminator edge and $0$ half-edges of color class $i$.
    Hence, if no half-edge is selected for color class $i$ in $G'$, we do not select an edge of color class $i$ in $G$ for $M$.
    Next, assume that exactly two half-edges are selected for color class $i$ in $G'$.
    Let $x, y \in V(G)$ be the vertices of $G$ corresponding to the vertices of $G'$ incident to these two half-edges. 
    By the second statement of \Cref{claim:structure:DCS}, if $M'$ contains $2$ half-edges of color class $i$, then these half-edges do not belong to the same non-trivial part $K_j^i$ from color class $i$.
    Hence, in $G$ there is an edge between $x$ and $y$ in color class $i$ and we can add this edge to $M$.
    Observe that $|M| = \frac{1}{2} |M'|$.
    To see that $M$ is feasible, observe that the degrees of the vertices in $G'$ that correspond to vertices in $G$ are the same in $M'$ and $M$. 
    Therefore, $M$ is a matching. To see that $M$ is  rainbow, observe that by construction from each color class at most one edge is selected in $M$.
    Hence, $M$ is a rainbow matching of size $|M| = \frac{1}{2} |M'|$.
\end{proof}

Since the construction of $(G', l', u')$ is polynomial and maximum $(l, u)$-matchings can be computed in polynomial time~\cite{schrijver}, \Cref{claim:DCS:quasiclique} implies \Cref{thm:dichotomy-result-rainbow-l-u-matchings}.

Notice that in any $\alpha$-CM-colored graph on $m$ edges, there is a polynomial number (at most $\binom{m}{r}$) of rainbow matchings on the at most $r$ color classes that are not CM.
By enumerating all those rainbow matchings and applying \cref{thm:dichotomy-result-rainbow-l-u-matchings} (letting $m_i = 1$ for each color class $i$) to the rest of the graph, we immediately obtain the polynomial part of the dichotomy.

\rmfrfree*
\begin{proof}
    Let $G$ be an $\alpha$-CM-colored graph. Furthermore, let $G_\alpha$ be the subgraph induced by those color classes that are not complete multipartite. Notice that we can decide in polynomial time whether a graph is CM.
    If $\alpha = 0$, then all color classes are complete multipartite.
    By~\cref{thm:dichotomy-result-rainbow-l-u-matchings}, the Rainbow Matching is polynomial on such graphs.
    Now assume $\alpha \geq 1$. 
    Since $\alpha$ is a constant (and hence finite), there exist $O(|E|^\alpha)$ many subsets of edges on the graph $G_\alpha$ that form a rainbow matching. 
    For each of these subsets, 
    we can 
    check in polynomial time whether they form a rainbow matching in $G_\alpha$.

    For each valid rainbow matching $M$ found this way, we remove the edge sets of the $\alpha$ color classes from $G$ and the vertex set incident to $M$.
    In the remaining graph $G'$, all color classes are complete multipartite, since the class of complete multipartite graphs is closed under vertex deletion. 
    By \cref{thm:dichotomy-result-rainbow-l-u-matchings}, we may compute a maximum rainbow matching $M'$ of $G'$ in polynomial time.
    The algorithm returns a rainbow matching $M \cup M'$ of $G$ of maximum cardinality.
    Since all rainbow matchings in $G_\alpha$ can be enumerated in polynomial time and extended to rainbow matchings of $G$ in polynomial time as well (notice that $M$ and $G'$ are vertex disjoint), the algorithm is polynomial.
\end{proof}

\section{Conclusion and Future Work}

We considered the problem \textsc{Maximum Rainbow Matching}, in which we are given an edge-colored graph and the task is to compute a matching of maximum cardinality, subject to picking at most one edge from each color class.
We provided a complexity dichotomy based on the structure of the color classes.
In particular, we showed that \textsc{Maximum Rainbow Matching} is polynomial-time solvable if each color class is a complete multipartite graph, and \NP-hard otherwise.
In light of our application in the introduction, a natural direction for future work is to consider generalizations of the problem to hypergraphs.

To the best of our knowledge, our result is the first complexity dichotomy for rainbow combinatorial optimization problems that is based on the structure of the color classes rather than on the structure of the input graph.
A further direction for future research is to consider other polynomial-time-solvable graph problems and derive analogous complexity dichotomies based on the structure of the color classes.
Typical examples include, for instance, rainbow variants of $s$--$t$ connectivity or arborescences.

Our main positive result on \textsc{Maximum Generalized Rainbow Matching} can be viewed as a matching problem with additional partition matroid constraints, and we give a polynomial-time algorithm whenever the edge set of each color class induces a complete multipartite graph.
We believe it would be interesting to explore settings in which the edges that may be selected from each color class are constrained by more general matroids.
A first step in this direction could be to consider laminar matroids.

\bibliography{references}

\end{document}